\documentclass[british,final]{lmcs}
\pdfoutput=1

\usepackage{lastpage}
\lmcsdoi{15}{1}{17}
\lmcsheading{}{\pageref{LastPage}}{}{}%
{Nov.~29,~2017}{Feb.~28,~2019}{}

\pdfoutput=1

\usepackage{xparse}
\usepackage{xpatch}
\usepackage{etoolbox}
\usepackage{environ}

\newcommand{\nocontentsline}[3]{}
\newcommand{\tocless}[2]{\bgroup\let\addcontentsline=\nocontentsline#1{#2}\egroup}

\makeatletter
\newcommand{\IfLabelExistsTF}[3]{\@ifundefined{r@#1}{#3}{#2}}
\makeatother

\usepackage[utf8]{inputenc}

\usepackage[british]{babel}
\usepackage{csquotes}
\usepackage[all]{foreign}
\defasforeign[mutatismutandis]{mutatis mutandis}
\defasforeign[Mutatismutandis]{Mutatis mutandis}
\defasforeign[criterium]{criterium}
\defnotforeign[wrt]{{\xperiodafter{w.r.t}}}

\usepackage{amsmath}
\usepackage{amsthm}
\usepackage{amsfonts}
\usepackage{amssymb}
\usepackage{stmaryrd}
\usepackage{centernot}
\usepackage{mathtools}

\allowdisplaybreaks

\usepackage{thmtools}

\makeatletter
\renewcommand{\xRightarrow}[2][]{\ext@arrow 03{10}{10}\Rightarrowfill@{#1}{#2}}
\makeatother

\usepackage{float}
\usepackage[section]{placeins}

\usepackage{tikz}
\usetikzlibrary{3d,calc,arrows,positioning,decorations.pathmorphing,shapes}

\tikzset{
	nat/.style={double,double equal sign distance},
	nat>/.style={nat,-implies},
	<nat/.style={nat,implies-},
	descr/.style={anchor=center,fill=white},
	crossing/.style={preaction={draw=white,-,line width=#1}},
	crossing/.default=2pt,
	extended/.style={shorten >=-#1, shorten <=-#1},
	extended/.default=2pt,
	text decoration line/.style={
		line width=.1ex,solid,
		rounded corners=0,
		round cap-round cap},
	capped line/.style={round cap-round cap},
	diagram/.style={
		font=\small,auto,scale=1.5,-latex,
		baseline=(current bounding box.center)}
}

\usepackage{hyperref}
\usepackage{nameref}

\usepackage[capitalise,nameinlink,noabbrev]{cleveref}

\makeatletter
\newcommand{\customlabel}[4][0]{%
	\protected@write\@auxout{}{\lstring\newlabel{#3}{{#4}{\thepage}{#4}{#3}{}}}%
	\protected@write\@auxout{}{\lstring\newlabel{#3@cref}{{[#2][#1][#1]#4}{\thepage}}}%
}
\makeatother

\hypersetup{
	hidelinks,
	final,
}

\usepackage[
	sectionbib,
	square,
	numbers,
	sort&compress
	]{natbib}
\AtBeginDocument{\bibliographystyle{abbrvnat}}

\theoremstyle{plain}
\newtheorem{theorem}{Theorem}[section]
\newtheorem{corollary}[theorem]{Corollary}
\newtheorem{proposition}[theorem]{Proposition}
\newtheorem{lemma}[theorem]{Lemma}

\theoremstyle{definition}
\newtheorem{definition}[theorem]{Definition}

\newtheorem{remark}[theorem]{Remark}
\newtheorem{example}[theorem]{Example}

\crefname{proof}{proof of}{proofs of}
\Crefname{proof}{Proof of}{Proofs of}

\makeatletter
\newcounter{proofatend@statemet}

\newcommand*\fixstatement[2][Proof of]{%
	\expandafter\preto\csname end#2\endcsname{%
		\global\def\proofatend@proofof{#1}%
		\stepcounter{proofatend@statemet}%
		\label{proofatend:\theproofatend@statemet}%
}}

\def\proofatend@toks{1000}
\newcounter{proofatend@stored}
\newcounter{proofatend@printed}

\let\proofatend@omit0
\let\proofatend@inplace1
\let\proofatend@postpone2

\newcommand{\omitproofs}{\global\let\proofatend@mode\proofatend@omit}
\newcommand{\postponeproofs}{\global\let\proofatend@mode\proofatend@postpone}
\newcommand{\keepproofs}{\global\let\proofatend@mode\proofatend@inplace}

\postponeproofs

\NewEnviron{proofatend}[1][Proof]{\ignorespaces%
	\if\proofatend@mode\proofatend@inplace%
		\begin{proof}[#1]\BODY\end{proof}%
	\else\if\proofatend@mode\proofatend@postpone%
		\edef\next{%
			\noexpand\begin{proof}[{\proofatend@proofof~\noexpand\cref{proofatend:\theproofatend@statemet}}]%
			\noexpand\phantomsection%
			\noexpand\label{proofatend:proof-\theproofatend@statemet}%
			\unexpanded\expandafter{\BODY}}%
		\global\toks\numexpr\proofatend@toks+\value{proofatend@stored}\relax=\expandafter{\next\end{proof}}%
		\stepcounter{proofatend@stored}%
	\fi\fi\ignorespacesafterend}

\NewDocumentCommand{\omittedproofs}{o}{%
	\ifnum\value{proofatend@printed}<\value{proofatend@stored}%
		\IfNoValueF{#1}{#1}%
		\count@=\value{proofatend@printed}%
		\loop%
			\the\toks\numexpr\proofatend@toks+\count@\relax%
			\stepcounter{proofatend@printed}%
			\ifnum\count@<\value{proofatend@stored}%
				\advance\count@\@ne%
		\repeat%
	\fi%
}
\makeatother

\AtEndPreamble{
	\fixstatement{theorem}
	\fixstatement{proposition}
	\fixstatement{corollary}
	\fixstatement{lemma}
	\fixstatement{fact}
  \fixstatement[Solution to]{exercise}
}

\DeclareMathOperator{\supp}{supp}

\newcommand{\defeq}{\triangleq}
\newcommand{\defiff}{\stackrel{\triangle}{\iff}}

\newcommand{\To}{\Rightarrow}

\newcommand{\cat}[1]{\textnormal{\bfseries #1}\xspace}
\newcommand{\Cat}{\cat{Cat}}
\newcommand{\Set}{\cat{Set}}
\newcommand{\kl}{\cat{Kl}}
\newcommand{\Fun}[2]{\left[#1,#2\right]}
\newcommand{\JFun}[3]{{\Fun{#1}{#2}^{#3}}}

\newcommand{\Cpoj}{\ensuremath{\omega\cat{-Cpo}^{\!\vee\!}}}
\newcommand{\Dcpo}{\ensuremath{\cat{DCpo}}}
\newcommand{\Dcpoj}{\ensuremath{\cat{DCpo}^{\!\vee\!}}}
\newcommand{\Sup}{\cat{Sup}}

\newcommand{\Id}{{I\mspace{-1mu}d}}
\newcommand{\UM}[2][\Sigma]{{#2}^{#1}}

\newcommand{\UAM}[2][\Sigma]{{#2}^{#1}_{\scriptstyle \checkmark}}

\newcommand{\LTS}[1][\Sigma]{\UM[#1]{\mathcal{P}}}
\newcommand{\ENA}[1][\Sigma]{\UAM[#1]{\mathcal{P}}}

\newcommand{\CM}{{\mathcal{C\mspace{-2mu}M}}}
\newcommand{\ST}[1][\Sigma]{\UM[#1]{\CM}}
\newcommand{\SA}[1][\Sigma]{\UAM[#1]{\CM}}

\newcommand{\lstr}{\mathit{\thickmuskip=0mu lstr}}
\newcommand{\rstr}{\mathit{\thickmuskip=0mu rstr}}
\newcommand{\dstr}{\mathit{\thickmuskip=0mu dstr}}

\makeatletter
\def\liftKl#1{{\mathpalette\liftKl@{#1}}}
\def\liftKl@#1#2{{%
\tikz[baseline=(n.base)]{
		\node[inner sep=1pt,outer sep=0pt] (n) {$\m@th#1#2$};
		\draw[text decoration line, capped line] ($(n.north west)+(.1ex,.0ex)$) -- ($(n.north west)+(.2ex,.1ex)$) -- ($(n.north east)+(-.2ex,.1ex)$) -- ($(n.north east)+(-.1ex,.0ex)$);
}}}

\newcommand*{\nxrightarrow}[2]{\xrightarrow{(#1,\,#2)}}
\newcommand*{\nxRightarrow}[2]{\xRightarrow{(#1,\,#2)}}
\newcommand*{\xrsquigarrow}[1]{\mathrel{%
\begin{tikzpicture}[
	baseline= {( $ (current bounding box.south) + (0,-0.5ex) $ )}]
	\node[inner sep=.5ex,minimum width=2ex] (lbl) {$\scriptstyle #1$};
\path[draw,<-,decorate,
  decoration={zigzag,amplitude=0.7pt,segment length=1.2mm,pre=lineto,pre length=4pt}] 
    (lbl.south east) -- (lbl.south west);
\end{tikzpicture}}%
}

\begin{document}

\title[Timed behavioural equivalences]{Behavioural equivalences for timed systems}

\author[Tomasz Brengos]{Tomasz Brengos\rsuper{a}}
\address{%
	\lsuper{a}Faculty of Mathematics and Information Science,
	Warsaw University of Technology,
	Koszykowa~75, 00-662 Warszawa, Poland
}
\email{t.brengos@mini.pw.edu.pl}
\thanks{Supported by the grant of Warsaw University of Technology no.~504M for young researchers.}

\author[Marco Peressotti]{Marco Peressotti\rsuper{b}}
\address{%
	\lsuper{b}Department of Mathematics and Computer Science,
	University of Southern Denmark,
	Campusvej~55, DK-5230 Odense M, Denmark
}
\email{peressotti@imada.sdu.dk}
\thanks{Partially supported by the Independent Research Fund Denmark, Natural Sciences, grant DFF-7014-00041.}

\begin{abstract}
  \looseness=-1
  Timed transition systems are behavioural models that include an explicit treatment of time flow and are used to formalise the semantics of several foundational process calculi and automata. 
  Despite their relevance, a general mathematical characterisation of timed transition systems and their behavioural theory is still missing.
  We introduce the first uniform framework for timed behavioural models that encompasses known behavioural equivalences such as timed bisimulations, timed language equivalences as well as their weak and time-abstract counterparts. All these notions of equivalences are naturally organised by their discriminating power in a spectrum. We prove that this result does not depend on the type of the systems under scrutiny: it holds for any generalisation of timed transition system. We instantiate our framework to timed transition systems and their quantitative extensions such as timed probabilistic systems. 
\end{abstract}

\maketitle

%-------------------------------------------------------------------------------
%	MAIN MATTER
%-------------------------------------------------------------------------------

\section{Introduction}
\label{sec:introduction}

Since Aczel’s seminal work \cite{am89:final}, the theory of coalgebras has been recognised as a good context for the study of concurrent and reactive systems \cite{rutten:universal}: systems are represented as maps of the form $X\to BX$ for a suitable \emph{behavioural functor} $B$. By changing the underlying category and functor a wide range of cases are covered, from traditional LTSs to systems with I/O, quantitative aspects, probabilistic distribution, and even systems with continuous state. 
Frameworks of this kind provide great returns from a theoretical and a practical point of view, since they prepare the ground for general results and tools which can be readily instantiated to various cases, and  they help us discover connections and similarities between apparently different notions.
Among the several valuable results offered by the coalgebraic approach we mention general accounts of bisimulation \cite{am89:final,staton11}, structural operational semantics \cite{tp97:tmos,klin:tcs2011,mp:qapl14,mp2016:tcs-gsos,ks2013:w-s-gsos}, weak bisimulation \cite{mp2013:weak-arxiv,gp:icalp2014,brengos2014:cmcs,brengos2015:jlamp}, trace equivalence \cite{hasuo07:trace,jacobssilvasokolova2012:cmcs,peressotti:phdthesis,kerstan2013coalgebraic}, minimization \cite{bonchi2014:tcl-mini}, determinisation \cite{bonchi2013:lmcs-det}, up-to techniques \cite{bonchi2014:lics-upto}, encodings \cite{mp:ictcs2016}.

Timed transition systems (TTSs) are behavioural models that include an explicit treatment of time flow and are used to formalise the semantics of several foundational process calculi and automata \cite{wang90,schneider:1995,nicolin1991,hooman:1992,yi:1990,alur:tcs1994-at}.
In \cite{kick:coalgebraic_semantics}, \citeauthor{kick:coalgebraic_semantics} introduced the first coalgebraic characterisation of (strong) timed bisimulation for \citeauthor{wang90}'s Timed CCS \cite{wang90} and similar calculi with explicit time (\eg delay operations, timeouts) and instantaneous actions (\eg process communication). Key to this characterisation is the introduction of the \emph{evolution comonad} $\mathcal{E}$: a comonad whose coalgebras capture state changes due to the passage of time in the sense that $\mathcal{E}$-coalgebras are (isomorphic to) partial left actions of the monoid modelling time (\eg $([0,\infty),+,0)$).
If a behavioural functor $B\colon\Set \to \Set$ admits a cofree comonad $B^\infty$ then, the timed extension of behaviours modelled as $B$-coalgebras is modelled as coalgebras for the product of comonads $\mathcal{E}\times B^\infty$ \cite{kick:coalgebraic_semantics}.
This construction applies only to the subclass of TTSs that separate timed transitions from discrete actions (a transition either models a change due to the passage of time or change due to an instantaneous interaction). Therefore, \cite{kick:coalgebraic_semantics} does not provide a model of Timed CSP \cite{schneider:1995} or Timed Automata \cite{alur:tcs1994-at}.

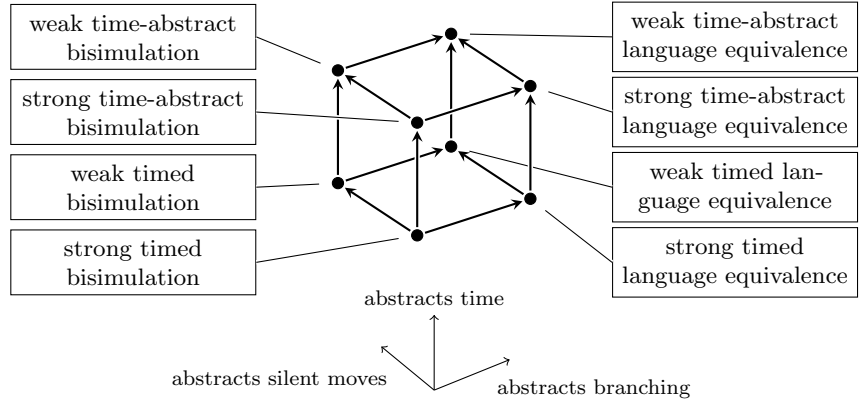
\begin{figure}
	\centering
	\begin{tikzpicture}[
			auto, font=\footnotesize,
			scale=1,		
			box/.style={text width=#1, align=center,draw,fill=white},
			dot/.style={circle,draw,fill,inner sep=1.5pt,outer sep=1pt},
			box con/.style={-,thin,shorten >=4pt,crossing=1.7pt},
			hasse con/.style={-stealth,thick,crossing=2pt}
		]

		\node[box={3cm}] (b000) at (-4,-1.5) {strong timed bisimulation};
		\node[box={3cm}] (b100) at (-4,-0.5) {weak timed bisimulation};
		\node[box={3cm}] (b010) at (-4, 0.5) {strong time-abstract bisimulation};
		\node[box={3cm}] (b110) at (-4, 1.5) {weak time-abstract bisimulation};
		\node[box={3cm}] (b001) at ( 4,-1.5) {strong timed language equivalence};
		\node[box={3cm}] (b101) at ( 4,-0.5) {weak timed language equivalence};
		\node[box={3cm}] (b011) at ( 4, 0.5) {strong time-abstract language equivalence};
		\node[box={3cm}] (b111) at ( 4, 1.5) {weak time-abstract language equivalence};

		\begin{scope}[
				yshift=.2cm,
			]
		\begin{scope}[
			z={(0.866cm,0.4cm)}, 
			x={(-0.866cm,0.4cm)}, 
			y={(0cm,1cm)},
			rotate around y=-10,
			scale=.75,
		]
			\foreach \x/\xc in {0/-1,1/1}{
				\foreach \y/\yc in {0/-1,1/1}{
					\foreach \z/\zc in {0/-1,1/1}{
					\node[dot] (p\x\y\z) at (\xc,\yc,\zc) {};
			}}}
		\end{scope}
		\end{scope}

		\begin{scope}[
				yshift=-3.2cm,
			]
		\begin{scope}[
				z={(0.866cm,0.5cm)}, 
				x={(-0.866cm,0.5cm)}, 
				y={(0cm,1cm)},
				rotate around y=-10,
				scale=1,
				font=\scriptsize,
				auto,
				-stealth
			]
			\draw[->] (0,0,0) 
				to node [pos=.7,anchor=north east,] 
				{abstracts silent moves}
				(1,0,0);
			\draw[->] (0,0,0) 
				to node [pos=1,anchor=south,]
				{abstracts time}
				(0,1,0);
			\draw[->] (0,0,0) 
				to node [pos=.7,anchor=north west,swap]
				{abstracts branching}
				(0,0,1);
		\end{scope}
		\end{scope}

		\foreach \x in {0,1}{
			\foreach \y in {0,1}{
				\foreach \z/\a in {0/east,1/west}{
			\draw[box con] (b\x\y\z.\a) -- (p\x\y\z);
		}}}
				
		\draw[hasse con] (p100) to (p101);
		\draw[hasse con] (p001) to (p101);
		\draw[hasse con] (p101) to (p111);

		\draw[hasse con] (p000) to (p100);
		\draw[hasse con] (p000) to (p010);
		\draw[hasse con] (p000) to (p001);
		\draw[hasse con] (p100) to (p110);
		\draw[hasse con] (p010) to (p110);
		\draw[hasse con] (p010) to (p011);
		\draw[hasse con] (p001) to (p011);
		\draw[hasse con] (p011) to (p111);
		\draw[hasse con] (p110) to (p111);
		
		\draw[box con] (b010.east) -- (p010); %redraw

	\end{tikzpicture}
	\caption{The spectrum of branching and linear equivalences of timed systems.}
	\label{fig:timed-equivalences-spectrum}
\end{figure}

In this paper, we introduce the first categorical framework for timed behavioural models like TTSs and their rich behavioural theory.
We provide a new definition of behavioural equivalence called  \emph{$q$-bisimulation} building on the theory of  bisimulation and a new extension of saturation \cite{brengos2015:corr,brengos2015:jlamp,brengos2015:lmcs} we introduce in this work to support time and time abstraction. 
The definition is parametrised and these parameters drive the saturation component describing how computations can be observed \eg whether the duration of single steps is observed, their combined duration, or no duration at all.
For the first time, we are able to capture many behavioural equivalences of interest: we show that the notions of timed bisimulation, timed language equivalence, as well as their weak and time-abstract counterparts, all correspond to specific instances of $q$-bisimulation.

All these notions of behavioural equivalence for TTSs are known to have different discriminating power and are naturally organised in the expressiveness spectrum reported in \cref{fig:timed-equivalences-spectrum} as a Hasse diagram (more discriminating notions are at the bottom and less ones at the top).
We present general result for deriving this kind of expressiveness spectra by simply looking at the parameters used to instantiate the abstract definition of $q$-bisimulation.
Furthermore, we show that these spectra are independent from specific computational effects: they hold for (non-deterministic) TTS as well as their quantitative extensions like timed probabilistic systems.

\subsection*{Synopsis and related work}
This work is closely related to research presented in \cite{brengos2014:cmcs,brengos2015:lmcs,brengos2015:jlamp,brengos2015:corr} with the emphasis laid on \cite{brengos2015:corr}. Indeed, \citeauthor{brengos2015:corr}' \cite{brengos2015:corr}, which is highly motivated by \citeauthor{sobocinski:jcss}'s work on relational presheaves (\ie lax functors whose codomain category is the category of sets and relations) and their saturation \cite{sobocinski:jcss}, presents the lax functorial framework as a natural extension of weak bisimulation via saturation studied in \cite{brengos2015:lmcs,brengos2015:jlamp}. The main focus of \cite{brengos2015:corr} is on lax functorial weak bisimulation and reflexive and transitive saturation. We remark that the first author already pointed out in \loccit that timed transition systems and their weak behavioural equivalence can be modelled in the lax functorial setting. This work extends these results in a systematic way by:
\begin{itemize}
  \item presenting the concept of general saturation and the family of behavioural equivalences associated with it (\cref{sec:general-saturation-and-equivalences});
	\item describing the categorical framework for behavioural models with explicit time treatment (\cref{sec:systems-as-functors});
  \item capturing a much wider spectrum of language and behavioural equivalences (\cref{sec:timed-equivalences});
  \item providing new case studies of timed behavioural models like \eg Segala systems and weighted systems (\cref{sec:more-applications}).
\end{itemize}
Preliminaries on the behavioural theory of timed transition systems and the necessary categorical machinery are in \cref{sec:timed-transition-systems,sec:categorical-background}, respectively.
Final remarks are in \cref{sec:conclusion}.
This paper is an extended and improved version of our conference paper \cite{brengos2016:concur}. We included new results on general expressiveness spectra for behavioural equivalences and instantiated our framework on Segala and other quantitative systems.

\section{Timed transition systems and behavioural equivalences}
\label{sec:timed-transition-systems}

The aim of this section is to recall basic notions on timed transition systems (TTS) and their behavioural equivalences known in the classical literature.
Timed transition systems are, in their most general form, labelled transition systems whose transitions are labelled with action symbols and time durations.
They are used as semantics models of timed automata (\cf \cite{alur:tcs1994-at}) and timed processes (\cf \cite[Sec.~2.3]{kick:phdthesis}).

In the sequel, we write $\Sigma$ for set of action symbols and $\mathbb{T}$ for the set of time durations. We reserve the symbol $\tau$ to transitions meant to be unobservable and write $\Sigma_\tau$ for the set $\Sigma + \{\tau\}$ of action symbols extended with $\tau$.
We assume that the set of time durations carries a monoid structure which we denote as $\mathbb{T}$ unless otherwise specified; the prototypical example is the monoid $([0,\infty),+,0)$ of positive real numbers under addition.

\subsection{Timed behavioural equivalences}
Fix a TTS $\alpha\colon X \to \mathcal{P}(\Sigma_\tau \times \mathbb{T} \times X)$. We write $x \nxrightarrow{\sigma}{t} y$ to denote a timed step in $\alpha$ meaning that $y \in \alpha(x)(t,\sigma)$.
We write $x \nxRightarrow{\sigma}{t} y$ for a saturated timed step in $\alpha$ where $\To$ denotes the least relation closed under the following:
\begin{equation*}
	\frac{}{x \nxRightarrow{\tau}{0} x}
	\qquad
	\frac{x \nxRightarrow{\tau}{t_0} x' \quad x' \nxrightarrow{\sigma}{t_1} y' \quad y' \nxRightarrow{\tau}{t_2} y \quad t = t_0+t_1+t_2}{x \nxRightarrow{\sigma}{t} y}
\end{equation*}

\begin{defiC}[\protect\cite{alur:tcs1994-at,larsen:tcs1997}]
	\label{def:nd-ta-behavioural-equivalences}
	For a TTS $\alpha$ and an equivalence relation $R$ on its carrier:
	\begin{itemize}
	\item 
		$R$ is a \emph{timed bisimulation} for $\alpha$ if $x \mathrel{R} y$ and $x\nxrightarrow{\sigma}{t}x'$ implies that there is $y'\in X$ such that $y\nxrightarrow{\sigma}{t} y'$ and $ x' \mathrel{R} y'$; 
	\item 
		$R$ is a \emph{(strong) time-abstract bisimulation} for $\alpha$ if $x \mathrel{R} y$ and $x\nxrightarrow{\sigma}{t}x'$ implies that there are $y'\in X$ and $t'\in [0,\infty)$ such that $y\nxrightarrow{\sigma}{t'} y'$ and $x' \mathrel{R} y'$; 
	\item
		$R$ is a \emph{weak timed bisimulation} for $\alpha$ if $x \mathrel{R} y$ and $x\nxRightarrow{\sigma}{t} x'$ implies that there is $y'\in X$ such that $y\nxRightarrow{\sigma}{t} y'$ and $x' \mathrel{R} y'$;  
	\item
		$R$ is a \emph{weak time-abstract bisimulation} for $\alpha$ if $x  \mathrel{R}  y$ and $x\nxRightarrow{\sigma}{t}x'$ implies that there are $y'\in X$ and $t'\in [0,\infty)$ such that $y\nxRightarrow{\sigma}{t'} y'$ and $x' \mathrel{R} y'R$.  
	\end{itemize}
\end{defiC}

Assume that $\alpha$ has also accepting (timed) steps \ie that it has type $\alpha \colon X \to \mathcal{P}(\Sigma_\tau \times \mathbb{T} \times X + \{\checkmark\})$. We write $x \xrightarrow{t} \checkmark$ to denote that $x$ can make an accepting move in time $t$ and terminate.
We write $x \xRightarrow{t} \checkmark$ for the saturated equivalent \ie any step derivable by means of the rule below.
\[
	\frac{
		x \nxRightarrow{\tau}{t} x' 
		\qquad
		x' \xrightarrow{t'} \checkmark
	}{
		x \xRightarrow{t + t'} \checkmark
	}
\]
A timed word $t_0\sigma_1t_1\dots\sigma_nt_n$ is accepted by a state $x_0$ provided there is a sequence of timed steps $x_0 \nxrightarrow{\sigma_1}{t_0}\ldots \nxrightarrow{\sigma_n}{t_{n-1}} x_n \xrightarrow{t_n} \checkmark$ in $\alpha$.
The timed language accepted by $x_0$ is the set 
\begin{align*}
	\mathrm{tl}_{\alpha}(x_0) \defeq &
	\{%\left\{ 
		t_0\sigma_1t_1\dots\sigma_nt_n \in \mathbb{T} \times (\Sigma_\tau\times \mathbb{T})^\ast
	\mid%\,\middle|\, 
		x_0 \nxrightarrow{\sigma_1}{t_0}\ldots \nxrightarrow{\sigma_n}{t_{n-1}} x_n \xrightarrow{t_n} \checkmark
	\}%\right\}
	\\
\intertext{of timed words accepted by $x_0$ and the untimed one is the set}
	\mathrm{utl}_{\alpha}(x_0) \defeq &
	\{%\left\{ 
		 	\sigma_1\ldots\sigma_n \in  \Sigma_\tau ^\ast  
	\mid%\,\middle|\, 
		 	\exists t_0,\ldots t_n \in \mathbb{T} \text{ s.t. } t_0\sigma_1t_1\dots\sigma_nt_n\in \mathrm{tl}_{\alpha}(x_0)
	\}%\right\}
\intertext{of untimed words accepted by $x_0$. A weaker notion is readily obtained by allowing saturated steps in the above definitions:}
	\mathrm{wtl}_{\alpha}(x_0) \defeq &
	\{%\left\{ 
		t_0\sigma_1t_1\dots\sigma_nt_n \in \mathbb{T} \times (\Sigma_\tau\times \mathbb{T})^\ast
	\mid%\,\middle|\, 
		x_0 \nxRightarrow{\sigma_1}{t_0}\ldots \nxRightarrow{\sigma_n}{t_{n-1}} x_n \xRightarrow{t_n} \checkmark
	\}%\right\}
	\\
	\mathrm{wutl}_{\alpha}(x_0) \defeq &
	\{%\left\{ 
		 	\sigma_1\ldots\sigma_n \in  \Sigma_\tau ^\ast  
	\mid%\,\middle|\, 
		 	\exists t_0,\ldots t_n \in \mathbb{T} \text{ s.t. } t_0\sigma_1t_1\dots\sigma_nt_n\in \mathrm{tl}_{\alpha}(x_0)
	\}%\right\}
	\text{.}
\end{align*}
\begin{definition} %[\protect\cite{alur:tcs1994-at}]
	\label{def:nd-ta-language-equivalences}
	For a TTS $\alpha$ and an equivalence relation $R$ on its carrier:
	\begin{itemize}
	\item 
		$R$ is a \emph{timed language equivalence} for $\alpha$ if $x \mathrel{R} y$ implies that $\mathrm{tl}_{\alpha}(x) = \mathrm{tl}_{\alpha}(y)$; 
	\item 
		$R$ is a \emph{time-abstract language equivalence} for $\alpha$ if $x \mathrel{R} y$ implies that $\mathrm{utl}_{\alpha}(x) = \mathrm{utl}_{\alpha}(y)$; 
	\item
		$R$ is a \emph{weak timed language equivalence} for $\alpha$ if $x \mathrel{R} y$ implies that $\mathrm{wtl}_{\alpha}(x) = \mathrm{wtl}_{\alpha}(y)$; 
	\item
		$R$ is a \emph{weak time-abstract language equivalence} for $\alpha$ if $x \mathrel{R} y$ implies that $\mathrm{wutl}_{\alpha}(x) = \mathrm{wutl}_{\alpha}(y)$. 
	\end{itemize}
\end{definition}

\subsection{Timed automata and their semantics}
\label{sec:timed-basic}

Timed automata are presented as machines akin to non-deterministic automata and equipped with a (finite) set of clocks for recording time flow. Besides consuming an input character, transitions have the side effect of resetting some of these clocks. Moreover, transition activation depends on the values stored in by the automata clocks and are conveniently described by means of \emph{guards} (or clock \emph{constraints}) \ie syntactic expressions generated by the grammar:
\begin{equation}
	\label{eq:clock-grammar}
	\delta \Coloneqq c\leq r\mid r\leq c \mid \neg \delta \mid \delta \wedge \delta
\end{equation}
where $c$ ranges over a given set of clocks $C$ and $r$ is a non-negative rational number (\cf \cite[Def.~3.6]{alur:tcs1994-at}). In the following, we write $\mathcal{G}(C)$ for the set of guards defined by \eqref{eq:clock-grammar} on the (finite) set of clocks $C$.
Then, timed automata are described by directed (multi) graphs whose edges are labelled with an input character, a clock guard, and a set of clocks to be reset to $0$.
\begin{definition}[Timed automaton]
	\label{def:timed-syntax}
	A \emph{timed automaton} or a \emph{timed transition table} is a tuple $\mathcal{A} = (\Sigma,L,C,E)$, where
	\begin{itemize}
	\item $\Sigma$ is a set called \emph{alphabet},
	\item $L$ is a set of \emph{states} or \emph{locations},
	\item $C$ is a set of \emph{clocks},
	\item $E\subseteq L\times \mathcal{G}(C)\times 2^C \times \Sigma \times L$ is a set of \emph{edges}.
	\end{itemize}
\end{definition}
It is important to note that the original definition of timed automaton and its semantics \cite{alur:tcs1994-at} come with an extra component that is missing in the above, namely an initial state. Given the aims of this work, this information can be safely omitted.

\begin{example}%[{\cite[Ex.~3.4]{alur:tcs1994-at}}]
	\label{ex:simple-ta}
	Consider the timed transition table from \cite[Ex.~3.4]{alur:tcs1994-at} depicted aside. The set
	\begin{minipage}[]{.8\textwidth}
		 $C$ of clocks $C$ is $\{c\}$ and the alphabet $\Sigma$ is $\{\sigma,\theta\}$. The edge from $l$ to $l'$ describes a transition that can be performed provided the input character is $\sigma$ and resets $c$ as its side effect. The other transition assumes $c < 2$ input $\theta$ and does not reset $c$.
	\end{minipage}\begin{minipage}[]{.2\textwidth}
		\hfill
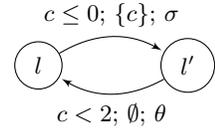
		\begin{tikzpicture}[
			font=\footnotesize,>=latex',shorten >=1pt,node distance=1.3cm,auto,
			state/.style={draw,circle,inner sep=4pt}
			]
			\node[state] (s0) {$l$};
			\node[state] (s1) [right=of s0] {$l'$};
			\draw[->,bend left] (s0) to node {$c \leq 0$; $\{c\}$; $\sigma$} (s1);
			\draw[->,bend left] (s1) to node {$c < 2$; $\emptyset$; $\theta$} (s0);
		\end{tikzpicture}
	\end{minipage}
\end{example}

\looseness=-1
Timed automata are abstract devices for recognising timed languages which in turn can be intuitively seen\footnote{Timed words are usually represented as particular words in $\Sigma \times [0,\infty)$ \ie sequences of pairs $(\sigma_i,t_i)$ where $i \leq j \implies t_i \leq t_j$ but these can be equivalently presented by means of time spans where $t_i$ is the time passed since the last transition.} as words over the alphabet $\Sigma \times [0,\infty)$. Then, it is natural to model the semantics of a timed automaton $\mathcal{A} = (\Sigma, L,C,E)$ as a labelled transition system $\alpha\colon S\to \mathcal{P}(\Sigma\times [0,\infty)\times S)$ where the state space $S$ actually is given by extending the set of locations $L$ with all possible assignments of values for its clocks. These assignments are called \emph{clock valuations} (or simply valuations) and are functions $v\colon C \to [0,\infty)$ from the set of clocks to the chosen domain of time---in the following, let $\mathcal{V}$ be the set of all valuations (for $C$). 

Before we formally define the semantics of timed automata, we need to define guard satisfiability and how resets and time flow affect clock valuations.
A valuation $v\colon C \to [0,\infty)$ is said to satisfy a guard $\delta \in \mathcal{G}(C)$, written $v \vDash \delta$, whenever the expression obtained by replacing in $\delta$ each clock with the value specified by $v$ holds. Formally:
\[\begin{array}{ll}
v \vDash c \leq r & \text{ if } v(c) \leq r\text{, }\\
v \vDash r \leq c & \text{ if } r \leq v(c)\text{, }\\
v \vDash \delta \land \delta' & \text{ if } v \vDash \delta \text{ and } v \vDash \delta'\text{, }\\
v \vDash \lnot\delta & \text{ if } v \not\vDash \delta\text{.}
\end{array}\]
For $v\colon C \to [0,\infty)$, $C'\subseteq C$, and $t\in [0,\infty)$ define $[C'\leftarrow 0]v$ and $v+t$ as follows:
\begin{align*}
	[C'\leftarrow 0]v(c) \defeq & {} \begin{cases}
		0 & \text{if } c \in C'\\
		v(c) & \text{otherwise}
	\end{cases}
	\\
	(v+t)(c) \defeq {} & v(c)+t
	\text{.}
\end{align*}

Let $t$ be the time elapsed since the automaton execution started and let $v$ be the evaluation that describes the value held by each clock. If there is an edge $(l,\delta,C',\sigma,l')$ in the transition table $E$ then, the automaton can  perform a transition from $l$ to $l'$ at time $t + t'$ provided that the input symbol is $\sigma$ and that the valuation $v+t'$ satisfies $\delta$. The transition has the effect of consuming the input symbol $\sigma$, changing the current location from $l$ to $l'$, and advancing all clocks of $t'$ units save for those in $C'$ which are reset to $0$. Formally:
\begin{definition}
	\label{def:timed-semantic-lts}
	For a timed automaton  $\mathcal{A} = (\Sigma, L,C,E)$ define its semantic model is the timed transition system $\alpha$ with state-space $L\times \mathcal{V}$ and transitions and transitions as follows
	\[
		(l,v) \xrightarrow{(\sigma,t)} (l',v') \defiff 
		\exists (l,\delta,C',\sigma,l')\in E \text{ s.t. } 
		v+t \vDash \delta \text{ and } 
		v' = [C'\leftarrow 0](v+t) \text{.}
	\]
\end{definition}

\begin{example}
	Recall the timed automaton described in Example~\ref{ex:simple-ta} and consider the LTS modelling its semantics. Traces starting in $(l,v)$ are sequences such that for any consecutive $(\sigma,t_i)$ and $(\theta,t_{i+1})$ we have $t_i + t_{i+1} < 2$.
\end{example}

\subsection{Timed processes and their semantics}
Timed processes calculi are a family of process calculi that include an explicit treatment of time flow as part of the behavioural model for the system under scrutiny; some illustrative examples are Timed CSP \cite{schneider:1995}, Timed CCS \cite{wang90}, ATP \cite{nicolin1991}.
Their semantics is given in terms of timed transition systems (in the sense of this section) however, some calculi restrict to special cases of TTSs to enforce certain assumptions on the system under scrutiny. 
For instance, Timed CCS and ATP processes distinguish between \emph{discrete}  and \emph{timed} transitions and define a TTS as a quadruple $(X,\Sigma_\tau,{\xrightarrow{}},{\xrsquigarrow{}})$ where
\begin{itemize}
	\item $X$ is the set of states;
	\item $\Sigma$ is the set of (action) symbols;
	\item ${\xrightarrow{}} \subseteq {X \times \Sigma_\tau \times X}$ represents ``duration-less'' \emph{discrete transitions}; 
	\item ${\xrsquigarrow{}} \subseteq {X \times [0,\infty)\times X}$ describes ``symbol-less'' \emph{time transitions}
\end{itemize}
and subject to 
\begin{align*}
	&x \xrsquigarrow{t+t'} x'\iff \exists x''\ x \xrsquigarrow{t} x'' \land x'' \xrsquigarrow{t'} x' 
	 & \qquad {(\emph{Continuity})}
	\\
	& x \xrsquigarrow{t} x' \land x \xrsquigarrow{t} x'' \implies x' = x''
	& \qquad {(\emph{Determinacy})}
\intertext{and, depending on the model, some additional constraints like \eg:}
	& x \xrsquigarrow{0} x
	& \qquad \text{(\emph{Zero Delay}\ \cite{larsen:tcs1997})}
	\\
	& x \xrightarrow{\tau} x' \implies x \centernot{\xrsquigarrow{t}}{}
	& \qquad \text{(\emph{Urgency}\ \cite{hooman:1992,yi:1990,nicolin1991})}
	\\
	& x \xrsquigarrow{t} x' \land x \xrightarrow{\sigma} x'' \implies x' \xrightarrow{\sigma} x''
	& \qquad \text{(\emph{Persistency}\ \cite{yi:1990})}
\end{align*}
A complete review of timed process calculi is out of the scope of this work; we refer the interested reader to \cite[Sec.~2.3]{kick:phdthesis}. 

A timed transition systems over the terms of a timed process calculus is then usually given by a SOS specification akin to classical, non-timed calculi.
A prototypical example of syntactic construct used by timed process calculi is the time prefix $(t).P$ that intuitively delays the execution of its continuation $P$ of $t$ units of time. The semantics of time prefixes can be specified by the following SOS rules.
\[
	\frac
		{}
		{(t).x \xrsquigarrow{t} x}
	\qquad
	\frac
		{}
		{(t+t').x \xrsquigarrow{t'} (t).x}
	\qquad
	\frac
		{x \xrsquigarrow{t'} x'}
		{(t).x \xrsquigarrow{t+t'} x'}
\]
Another example is the rule below which specifies that time affects components composed in parallel equally:
\[
	\frac
		{x_1 \xrsquigarrow{t} x'_1 \qquad x_2 \xrsquigarrow{t} x'_2}
		{x_1 \mid x_2  \xrsquigarrow{t} x_1' \mid x_2'}
\]

\section{Categorical background}\label{sec:categorical-background}

We assume that the reader is familiar with the following basic category theory notions: a category, a functor, a monad and an adjunction. Here we briefly recall some of them here and also present other basics needed in this paper.

\subsection{Coalgebras}
Let $\cat{C}$ be a category and $F \colon \cat{C}\rightarrow \cat{C}$ a functor.  An \emph{$F$-coalgebra} is a morphism $\alpha\colon X\to FX$ in $\cat{C}$.  The domain $X$ of $\alpha$ is called \emph{carrier} and the morphism  $\alpha$ is sometimes also called \emph{structure}. A \emph{homomorphism} from an $F$-coalgebra $\alpha\colon X\to FX$ to an $F$-coalgebra $\beta\colon Y\to FY$  is an arrow $f \colon X\rightarrow Y$ in $\cat{C}$ such that $ F(f)\circ \alpha =  \beta \circ f$. The category of all $F$-coalgebras and homomorphisms between them is denoted by $\cat{Coalg}(F)$. Many transition systems can be captured by the notion of coalgebra. The most important from our perspective are listed below. 

Let $\Sigma$ be a fixed set and put $\Sigma_\tau = \Sigma+\{\tau\}$. The label $\tau$ is considered a special label called \emph{silent} or \emph{invisible} label. 

\begin{example}[Labelled transition systems] $\mathcal{P}(\Sigma_\tau\times Id)$-coalgebras are  labelled transition systems over the alphabet $\Sigma_\tau$ \cite{milner:cc,rutten:universal, sangiorgi2011:bis}. Here, $\mathcal{P}$ denotes the powerset functor.  In this paper we also consider labelled transition systems with a monoid structure on labels, \ie coalgebras of the type $\mathcal{P}(M\times Id)$,  or even more generally, as coalgebras of the type $\mathcal{P}(\Sigma_\tau\times M\times Id)$ for a monoid $(M,\cdot,1)$. 
\end{example}

\begin{example}[Non-deterministic automata with $\varepsilon$-moves] $\mathcal{P}(\Sigma_\tau \times Id+1)$-coalgebras are non-deterministic automata with $\varepsilon$-transitions \cite{hasuo07:trace}. Here, $\varepsilon$-moves are labelled with $\tau$ and $1= \{\checkmark\}$ is responsible for specifying which states are final and which are not. To be more precise, given $\varepsilon$-NA $\alpha\colon X\to \mathcal{P}(\Sigma_\tau \times X +1)$ a state $x\in X$ is \emph{final} iff $\checkmark \in \alpha(x)$. 
\end{example}

\subsection{Monads and their Kleisli categories} 

A monad on a category \cat{C} is a triple $(T,\mu,\eta)$ where $T$ is an endofunctor over \cat{C} and $\mu \colon TT \To T$ and $\eta \colon \Id \To T$ are two natural transformations with the property that they make the diagrams below commute:
\[
	\begin{tikzpicture}[diagram]	
		\node (n0) at (0,1) {$T^3$};
		\node (n1) at (1,1) {$T^2$};
		\node (n2) at (0,0) {$T^2$};
		\node (n3) at (1,0) {$T$};
		\draw[nat>] (n0) to node {$\mu T$} (n1);
		\draw[nat>] (n0) to node[swap] {$T\mu$} (n2);
		\draw[nat>] (n1) to node {$\mu$} (n3);
		\draw[nat>] (n2) to node[swap] {$\mu$} (n3);
	\end{tikzpicture}
	\qquad
	\begin{tikzpicture}[diagram]	
		\node (n0) at (0,1) {$T$};
		\node (n1) at (1,1) {$T^2$};
		\node (n2) at (2,1) {$T$};
		\node (n3) at (1,0) {$T$};
		\draw[nat>] (n0) to node {$\eta T$} (n1);
		\draw[nat>] (n0) to node[swap] {$\Id$} (n3);
		\draw[nat>] (n2) to node[swap] {$T\eta$} (n1);
		\draw[nat>] (n1) to node {$\mu$} (n3);
		\draw[nat>] (n2) to node[] {$\Id$} (n3);
	\end{tikzpicture}	
\]
The natural transformations $\mu$ and $\eta$ are called \emph{multiplication} and \emph{unit} of $T$, respectively.

Each monad $(T,\mu,\eta)$ gives rise to a canonical category called \emph{Kleisli category of $T$} and denoted by $\kl(T)$. This category has the same objects of the category $\cat{C}$ underlying $T$; its hom-sets are given as $\kl(T)(X,Y) = \cat{C}(X,TY)$ for any two objects $X$ and $Y$ in $\cat{C}$ and its composition as
\begin{equation*}
	\begin{tikzpicture}[diagram]	

		\node[left] (n0) at (0,0) {$X$};

		\foreach \l/\n [count = \i] in {
			f/TY,
			Tg/TTZ,
			\mu_{Z}/TZ
		}{
			\pgfmathparse{int(\i-1)}
			\draw[->] (n\pgfmathresult.east) --++(.75,0)
					node[below,pos=.5] {$\l$}
					node[right] (n\i) {$\n$};			
		};
		
		\draw[->,rounded corners] (n0) -- +(0,.5) -| node[pos=.25] {$g \circ f$} (n\i);

	\end{tikzpicture}	
\end{equation*}
for any two morphisms $f$ and $g$ with suitable domain and codomain. 
There is an inclusion functor $(-)^\sharp\colon \cat{C} \to \kl(T)$ that takes every object $X$ to itself and every morphism $f\colon X \to Y$ to $\eta_Y \circ f$. This functor admits a right adjoint $U_T$ that takes every object $X$ to $TX$ and every morphism $f\colon X \to Y$ to its \emph{Kleisli extension} $\mu_Y \circ Tf$:
\begin{equation*}
	\begin{tikzpicture}[diagram,xscale=1.2]
		\node (n0) at (0,0) {$\cat{C}$};
		\node (n1) at (1,0) {$\kl(T)$};
		\draw[right hook->,bend left] (n0.north east) to node [above] {$(-)^\sharp$} (n1.north west);
		\draw[->,bend left] (n1.south west) to node [below] {$U_T$} (n0.south east);
		\node[rotate=-90] at ($(n0.east)!.5!(n1.west)$) {$\dashv$};
	\end{tikzpicture}	
\end{equation*}
This adjoint situation identifies the monad $T$ over $\cat{C}$.

A monad $(T,\mu,\eta)$ on a cartesian closed category $\cat{C}$ is called \emph{strong} if there is a natural transformation $\lstr_{X,Y}\colon X\times TY\to T(X\times Y)$ called \emph{tensorial strength} satisfying the strength laws listed in \eg \cite{kock1972strong}.
Existence of the transformation $\lstr$ is not a strong requirement. For instance all monads on $\Set$ are strong. 

\subsubsection*{Powerset monad}
The powerset endofunctor $\mathcal{P}\colon \Set\to \Set$ is a monad whose multiplication and unit are given on their $X$-components by: $\mu_X\colon\mathcal{P}\mathcal{P}X\to \mathcal{P}X; S \mapsto \bigcup S \text{ and }\eta_X\colon X\to \mathcal{P}X; x\mapsto \{x\}$.
The category $\kl(\mathcal{P})$ consists of sets as objects and maps of the form $X\to \mathcal{P}Y$ as morphisms. For $f\colon X\to \mathcal{P}Y$ and $g\colon Y\to \mathcal{P}Z$ the composition $g\circ f\colon X\to \mathcal{P}Z$ is as:
\[
g\circ  f(x) = \bigcup (\mathcal{P}g)(f(x))= \{z \mid z\in g(y) \text{ \& }y\in f(x) \text{ for some }y\in Y\}. 
\]
For any two sets $X,Y$ there is a bijective correspondence between maps $X\to \mathcal{P}Y$ and binary relations between elements of $X$ and $Y$. 
Indeed, for $f\colon X\to \mathcal{P}Y$ we put $R_f\subseteq X\times Y$, $(x,y)\in R_f \iff y\in f(x)$ and for $R\subseteq X\times Y$ we define $f_R\colon X\to \mathcal{P}Y;x\mapsto \{y\mid x R y\}$. It is now easy to see that the category $\kl(\mathcal{P})$ is isomorphic to the category $\cat{Rel}$ of sets as objects, binary relations as morphisms and relation composition as morphism composition.

\subsubsection*{LTS monad}
Labelled transition systems functor $\LTS = \mathcal{P}(\Sigma_\tau\times Id)$ carries a monadic structure 
$(\LTS,\mu,\eta)$ \cite{brengos2015:lmcs}, where the $X$-components of $\eta$ and $\mu$ are:
\[
\eta_X(x) = \{(\tau,x)\} \quad\text{ and }\quad \mu_X(S) = \bigcup_{(\sigma,S')\in S} \{(\sigma,x)\mid (\tau,x)\in S'\}\cup \bigcup_{(\tau,S')\in S} S'.
\]
For $f\colon X\to \LTS Y$ and $g\colon Y\to \LTS Z$ the composition $g\circ f$ in $\kl(\LTS)$ is
\[
(g\circ f)(x)= \{(\sigma,z)\mid x\xrightarrow{\sigma}_f y \xrightarrow{\tau}_g z \text{ or } x\xrightarrow{\tau}_f y \xrightarrow{\sigma}_g z\},
\]
where $x\xrightarrow{\sigma}_f y$ denotes $(\sigma,y)\in f(x)$. 

\begin{remark}
When seen as an object of $\kl(\mathcal{P})$, any set $\Sigma_\tau = \Sigma+\{\tau\}$ can be endowed with a monoid structure $(\Sigma_\tau,m,\tau)$ where the multiplication $m\colon \Sigma_\tau \times \Sigma_\tau \to \Sigma_\tau$ is defined as:
\[
  m(\sigma,\sigma') = \begin{cases}
  \{\sigma\} & \text{if } \sigma' = \tau\\
  \{\sigma'\} & \text{if } \sigma = \tau\\
  \emptyset & \text{otherwise}
  \end{cases}
\]
The $\LTS$ monad is equivalently defined as the composition of the canonical adjunction $(-)^\sharp \dashv U_{\mathcal{P}}$ and the writer monad for $(\Sigma_\tau,m,\tau)$.
\end{remark}

\subsubsection*{$\varepsilon$-NA monad}
As pointed out in \cite[Example 4.5]{brengos2015:lmcs} the $\varepsilon$-NA functor $\ENA = \mathcal{P}(\Sigma_\tau \times Id+1)$ carries a monadic structure $(\ENA,\mu,\eta)$. The composition in the Kleisli category for this monad is given as follows. For two morphisms $f\colon X\to \ENA Y$ and $g\colon Y\to \ENA Z$ their composition $g\circ f$ in $\kl(\ENA)$ is:
\begin{align*}
(g\circ f)(x) = & \{(\sigma,z)\mid x\stackrel{\sigma}{\to}_f y \xrightarrow{\tau}_g z \text{ or }x\xrightarrow{\tau}_f y \xrightarrow{\sigma}_g z\} \cup 
\\ &\{\checkmark \mid \checkmark\in f(x) \text{ or } x\xrightarrow{\tau}_f y \text{ and }  \checkmark\in g(y) \}.
\end{align*}

\subsection{Coalgebras with internal moves}\label{subsection:coalgebras_with_internal} Originally \cite{hasuo07:trace,silva2013:calco}, coalgebras with internal moves were introduced in the context of coalgebraic trace semantics as coalgebras of the type $T(F+Id)$ for a monad $T$ and an endofunctor $F$ on a common category. In \cite{brengos2015:lmcs} Brengos showed that given some mild assumptions on $T$ and $F$ we may introduce a monadic structure on $T(F+Id)$.\footnote{
  If $T$ and $F$ do not meet the required assumptions but $F$ admits a free monad $F^\ast$ and $F$ distributes over $T$ then, $T(F+Id)$ can be embedded into $TF^{*}$ which in turn can be given a monadic structure. 
  When both constructions are available, they are equivalent in the sense that they yield the same notion of weak bisimulation \cite{brengos2015:lmcs}.
} The LTS monad and $\varepsilon$-NA monad (as well as all examples in \cref{sec:more-applications}) arise by the application this construction.

The trick of modelling  the silent steps via a monad allows us \emph{not} to specify the internal moves explicitly. Instead of considering $T(F+Id)$-coalgebras we consider $T'$-coalgebras for a monad $T'$. Unless otherwise stated, we assume that all types of coalgebras considered here carry a monadic structure.

To give a $T$-coalgebra is to give an endomorphism in $\kl(T)$. We use this observation and present our results in as general setting as possible. Hence, we will replace $\kl(T)$ with an arbitrary  category $\cat{K}$ and work in the context of endomorphisms of $\cat{K}$ bearing in mind our prototypical example of $\cat{K}=\kl(T)$. 

\subsection{Order enriched categories}

An \emph{order enriched} category is a category whose hom-sets are additionally equipped with a partial order structure and whose composition preserves the structure in the following sense:
\[
	f \leq f' \implies g\circ f \leq g\circ f' \land f \circ h \leq f' \circ h
	\text{.}
\]
In this paper, the ordering will naturally stem from the type of computational effects under scrutiny and correspond to a notion of simulation between systems.
For instance, in the case of non-deterministic systems, the ordering is given by pointwise extension of the inclusion order imposed by $\mathcal{P}$.
In this setting, adding transitions to a system respects the ordering.

For some results, we will require that our categories admit arbitrary non-empty joins and that they are preserved by composition (on the left or on the right) with morphisms from a given subcategory---our prototypical scenario is a Kleisli category and its underlying category. We refer to this property as left and right distributivity \cite[Def.~2.10 and 3.2]{brengos2015:jlamp}.
\begin{definition}
	Let $\cat{K}$ be order enriched and let $J\colon \cat{J} \to \cat{K}$ exhibit $\cat{J}$ as a subcategory of $\cat{K}$.
	The category $\cat{K}$ is called \emph{$J$-right distributive} provided that it the following holds for any non-empty family of morphisms $\{g_i\}_{i\in I}$ and any morphism $f \in \cat{J}$ with suitable domain and codomain:
	\begin{equation}
		\tag{$J$-RD} \label{law:RD} 
		\left (\bigvee_i g_i\right ) \circ f = \bigvee_i g_i \circ f
		\text{.}
	\end{equation}
	The category $\cat{K}$ is called \emph{$J$-left distributive} provided that it the following holds for any non-empty family of morphisms $\{g_i\}_{i\in I}$ and any morphism $f \in \cat{J}$ with suitable domain and codomain:
	\begin{equation}
		\tag{$J$-LD} \label{law:J-LD} 
		f\circ \left (\bigvee_i g_i\right)  = \bigvee_i f\circ g_i
		\text{.}
	\end{equation}
\end{definition}
We call any $Id$-left distributive (resp. $Id$-right distributive) category simply left distributive (resp. right distributive).

Objects in the image of the powerset monad $\mathcal{P}$ are naturally endowed with an order structure given by subset inclusion. This order structure is extended to hom-sets of $\kl(\mathcal{P})$ in a pointwise manner. Formally, the order structure of any hom-set $\kl(\mathcal{P})(X,Y)$ is given on any $f$ and $g$ as:
\[
	f \leq g \defiff \forall x f(x) \subseteq g(x)
\]
and suprema of any non-empty family $\{f_i\}_{i \in I}$ as:
\[
	\bigvee_{i \in I} f_i \defeq \lambda x.\bigcup_{i \in I} f_i(x)
	\text{.}
\]
Composition in $\kl(\mathcal{P})$ is monotone and preserves suprema of non-empty families (see \eg \cite{brengos2015:jlamp,brengos2016:concur}).
It follows that for any $J\colon \cat{J} \to \kl(\mathcal{P})$ that exhibits $\cat{J}$ as a wide subcategory of $\kl(\mathcal{P})$ (\eg the inclusion $(-)^\sharp\colon\Set\to\kl(\mathcal{P})$), the category $\kl(\mathcal{P})$ is $J$-distributive.

\subsection{Saturation-based behavioural equivalences}
\label{sec:special-saturation-equivalences}
The notion of strong bisimulation has been well captured coalgebraically \cite{rutten:universal,staton11}. 
Weak bisimulation can be captured coalgebraically as the combination of kernel bisimulation and saturation \cite{brengos2015:jlamp}. We briefly recall the relevant definitions and refer the interested reader to \loccit for further details.

Following \cite{brengos2015:jlamp}, we consider a formulation of coalgebraic bisimulation called \emph{kernel bisimulation} \ie ``a relation which is the kernel of a compatible refinement system'' \cite{staton11}. Formally, a kernel bisimulation (herein just \emph{bisimulation}) for a coalgebra $\alpha\colon X\to TX$ is a relation $R\rightrightarrows X$ (\ie a jointly monic span) in $\cat{C}$ provided that there is a coalgebra $\beta\colon Y\to TY$ (the refinement system) and an arrow $f\colon X\to Y$ (called behavioural morphism) for which $R \rightrightarrows X$ is the kernel pair in $\cat{C}$ and subject to the coalgebra homomorphism compatibility condition:
\[
	Tf\circ \alpha = \beta \circ f
	\text{.}
\]
Since this identity can be restated in terms of composition in $\kl(T)$ as \[
	f^\sharp \circ \alpha = \beta \circ f^\sharp\text{,}
\]
we can generalize the definition of bisimulation to the setting of endomorphisms as follows.
Let $J\colon \cat{J} \to \cat{K}$ exhibit a wide subcategory of $\cat{K}$. 
For an endomorphism $\alpha\colon X\to X\in \cat{K}$ we say that a relation on $X$ in $\cat{J}$ is a \emph{(strong) bisimulation} for $\alpha$ if it is a kernel pair of an arrow $f\colon X\to Y \in \cat{J}$ for which there is $\beta\colon Y\to Y\in \cat{K}$ satisfying 
\[
	J(f)\circ \alpha = \beta \circ J(f)
	\text{.} 
\]
If we take the inclusion $(-)^\sharp\colon\Set \to \kl(\LTS)$ as $J$ then the notions of bisimulation for coalgebras and on endomorphisms coincide. 
If we take the inclusion $(-)^\sharp\colon\kl(\mathcal{P}) \to \kl(\ENA)$ as $J$ then the above captures (finite) trace equivalence.
Taking the source category of behavioural morphisms to be different from $\Set$, say $\cat{Rel} \cong \kl(\mathcal{P})$, is akin to the classical approach towards modelling coalgebraic finite trace equivalence in terms of bisimulation in the Kleisli categories for a monadic part $T$ of the type functor $TF$ \cite{hasuo07:trace,jacobssilvasokolova2012:cmcs} where behavioural morphisms have underlying maps from $\kl(T)$.

In \cite{brengos2015:lmcs,brengos2015:jlamp} we presented a common framework for defining weak bisimulation and weak trace semantics via saturation for coalgebras with internal moves which encompasses several well known instances of this notion for systems among which we find labelled transition systems and fully probabilistic systems. The basic ingredient of this setting was the notion of the aforementioned \emph{saturation} for order enriched categories. Henceforth, assume that $\cat{K}$ is order enriched. An endomorphism $\alpha$ in $\cat{K}$ is called \emph{saturated} whenever it holds that:
\[
	id\leq \alpha \qquad \alpha \circ \alpha \leq \alpha
	\text{.}
\]
Under mild conditions (see \cite{brengos2015:jlamp}), every endomorphism can be canonically saturated in the sense that assigning endomorphisms to their saturation identifies a reflection between certain categories of endomorphisms.
Define $\cat{End}_J^{\leq}({\cat{K}})$ as the category whose objects are endomorphisms in $\cat{K}$ and whose morphisms are arrows in the subcategory $J$ that form a lax commuting square between endomorphisms. Namely, a morphism from $\alpha\colon X\to X$ to $\beta\colon Y\to Y$ is any morphism $f$ such that:
\[
 	J(f)\circ \alpha \leq \beta\circ J(f)
 	\text{.}
\]
Define $\cat{End}_J^{\ast\leq}(\cat{K})$ as the full subcategory of $\cat{End}_J^{\leq}({\cat{K}})$ of saturated endomorphisms.
The category $\cat{K}$ is said to admit saturation (\wrt $J$) whenever it holds that:
\begin{equation}
	\begin{tikzpicture}[diagram]
		\node (n0) {$\cat{End}_J^{\leq}(\cat{K})$};
		\node[right=6ex of n0] (n1) {$\cat{End}_J^{\ast\leq}(\cat{K})$};
		\draw[->,bend left] (n0.north east) to node [above] {$(-)^\ast$} (n1.north west);
		\draw[left hook->,bend left] (n1.south west) to (n0.south east);
		\node[rotate=-90] at ($(n0.east)!.5!(n1.west)$) {$\dashv$};
	\end{tikzpicture}
	\label{eq:special-saturation-adjunction}
\end{equation}
In this setting we refer to $\alpha^\ast$ as \emph{the} saturation of $\alpha$.

A relation $R\rightrightarrows X$ in $\cat{J}$ is a \emph{weak bisimulation} on an endomorphism $\alpha\colon X\to X$ in $\cat{K}$ if it is a (strong) bisimulation for the saturation $\alpha^\ast\colon X\to X$ of $\alpha$ in $\cat{K}$. 
As will be witnessed in sections to come, different weak equivalences for timed systems will be defined in an analogous manner.

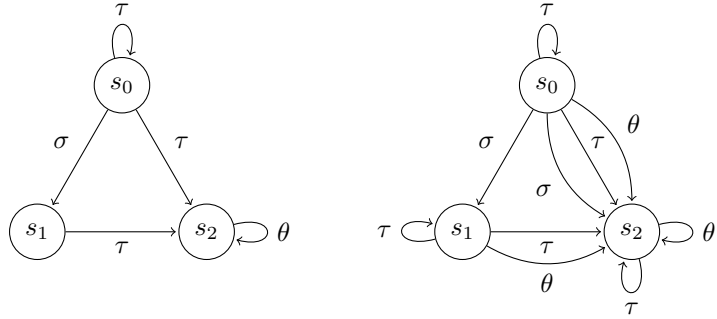
\begin{figure}
\centering
\begin{tikzpicture}[
  shorten >=1pt,auto,
  font=\small,
  state/.style={draw, circle},
  baseline=(current bounding box.north)
]
\node[state] (q_0) at ( 90:1.3) {$s_0$};
\node[state] (q_1) at (210:1.3) {$s_1$};
\node[state] (q_2) at (330:1.3) {$s_2$};
\draw[->] (q_0) to[loop above] node {$\tau$} (q_0);
\draw[->] (q_0) to node[swap] {$\sigma$} (q_1);
\draw[->] (q_0) to node {$\tau$} (q_2);
\draw[->] (q_1) to node[swap] {$\tau$} (q_2);
\draw[->] (q_2) to[loop right] node {$\theta$} (q_2);
\end{tikzpicture}
\qquad
\begin{tikzpicture}[
  shorten >=1pt,auto,
  font=\small,
  state/.style={draw, circle},
  baseline=(current bounding box.north)
]
\node[state] (q_0) at ( 90:1.3) {$s_0$};
\node[state] (q_1) at (210:1.3) {$s_1$};
\node[state] (q_2) at (330:1.3) {$s_2$};
\draw[->] (q_0) to[loop above] node {$\tau$} (q_0);
\draw[->] (q_0) to node[swap] {$\sigma$} (q_1);
\draw[->] (q_0) to node {\!\!$\tau$} (q_2);
\draw[->] (q_0) to[bend right] node[swap] {$\sigma$} (q_2);
\draw[->] (q_0) to[bend left] node {$\theta$} (q_2);
\draw[->] (q_1) to[loop left] node {$\tau$} (q_1);
\draw[->] (q_1) to node[swap] {$\tau$} (q_2);
\draw[->] (q_1) to[bend right] node[swap] {$\theta$} (q_2);
\draw[->] (q_2) to[loop right] node {$\theta$} (q_2);
\draw[->] (q_2) to[loop below] node {$\tau$} (q_2);
\end{tikzpicture}
\caption{A labelled transition system (left) and its saturation (right).}
\label{fig:lts-saturation-example}
\end{figure}
\begin{example}
Let $\alpha\colon X\to \LTS X$ be a labelled transition system reported in \cref{fig:lts-saturation-example} on the left. In this case, if we put \cat{K} to be Kleisli category for the LTS monad then the saturated map $\alpha^\ast$ is the LTS reported in \cref{fig:lts-saturation-example} on the right. Moreover, in general, if $\cat{K} = \kl(\LTS)$ and \cat{J} is put to be \cat{Set} then weak bisimulations coincide with Milner’s weak bisimulations for labelled transition systems \cite{milner:cc}.
\end{example}

\section{General saturation and behavioural equivalences}
\label{sec:general-saturation-and-equivalences}

The theory of saturation from the previous section offers a mechanism to compare system behaviours abstracting over certain patterns of computational steps. In this section we extend the theory of saturation to computations that may depend on aspects of their environment such as time; we refer to this conservative extension of (special) saturation as \emph{general saturation}. General saturation provides two orthogonal abstraction dimensions that may be arbitrarily combined: 
\begin{itemize}
	\item 
		General saturation allows to abstract from effects of the computation like unobservable steps (like special saturation);
	\item 
		General saturation allows to abstract from aspects of the computation environment like time flow.
\end{itemize}
As a consequence, a new orthogonal abstraction dimension is added to the spectrum of saturation behavioural equivalences as shown \eg in the spectrum of equivalences for timed systems in \cref{fig:timed-equivalences-spectrum}. A detailed construction of the spectrum is postponed to \cref{sec:timed-equivalences}.

In the sequel let $J\colon \cat{J} \to \cat{K}$ be a functor that exhibits $\cat{J}$ as a wide subcategory of $\cat{K}$ and let $\cat{K}$ be order enriched. 
As we will discuss in \cref{sec:timed-equivalences}, our prototypical example of $J\colon \cat{J} \to \cat{K}$ is the inclusion $(-)^\sharp\colon\Set \to \kl(\LTS)$  when interested in bisimulations for TTSs, $(-)^\sharp\colon\Set \to \kl(\ENA)$ when interested in bisimulations for TTSs with accepting moves, and $(-)^\sharp\colon\kl(\mathcal{P}) \to \kl(\ENA)$ when interested in language equivalences for TTSs.

\subsection{Lax functors}
A \emph{lax functor} from a category $\cat{D}$ to $\cat{K}$ is an assignment $\pi$ for objects and morphisms in $\cat{D}$ with the property that:
\begin{itemize}
	\item $id_{\pi D}\leq \pi(id_D)$ for any object $D$ in $\cat{D}$,
	\item $\pi(d_1)\circ \pi(d_2) \leq \pi(d_1\circ d_2)$ for any $d_1$ and $d_2$ in $\cat{D}$ with suitable domain and codomain.
\end{itemize}
Let $\pi$ and $\pi'$ be two lax functors from $\cat{D}$ to $\cat{K}$. A family $f=\{f_D\colon \pi D\to\pi'D\}_{D\in \cat{D}}$ of morphisms in $\cat{K}$ is called a \emph{lax natural transformation} from $\pi$ to $\pi'$ if the diagram below holds for any $d\colon D\to D'$ in $\cat{D}$.
\begin{equation*}
	\begin{tikzpicture}[diagram,xscale=1.2]
		\node (n0) at (0,1) {$\pi  D $};
		\node (n1) at (1,1) {$\pi' D $};
		\node (n2) at (0,0) {$\pi  D'$};
		\node (n3) at (1,0) {$\pi' D'$};
		\draw[->] (n0) to node {$f_{D}$} (n1);
		\draw[->] (n0) to node[swap] {$\pi(d)$} (n2);
		\draw[->] (n1) to node {$\pi'(d)$} (n3);
		\draw[->] (n2) to node[swap] {$f_{D'}$} (n3);
		\node[rotate=45] at ($(n0)!.5!(n3)$) {$\geq$};
	\end{tikzpicture}
\end{equation*}
\emph{Oplax functors} and \emph{oplax transformations} are defined by reversing the order in the definitions above. 
Note that in the more general 2-categorical setting an (op)lax functor and an (op)lax natural transformation are assumed to additionally satisfy extra coherence conditions \cite{lack:09}. In the setting of order enriched categories these conditions are vacuously true, hence we do not list them here.
Let $\cat{D}$ be a small category, we write $[\cat{D},\cat{K}]^J$ for the order enriched category defined by the following data:
\begin{itemize}
	\item lax functors from $\cat{D}$ to $\cat{K}$ as objects;
	\item oplax natural transformations with components laying in the image of $J$ (hereafter just ``from $J$'') as morphisms;
	\item the pointwise extension of the order enrichment of $\cat{K}$.
\end{itemize}
In particular, for any two morphism $f$ and $f'$ in $[\cat{D},\cat{K}]^J(\pi, \pi')$, it holds that $f\leq f'$ whenever $f_{D} \leq f'_D$ for any object $D$ of $\cat{D}$. Any functor $q\colon \cat{D}\to \cat{E}$ between small categories induces the \emph{change-of-base functor}
\[
	[q,\cat{K}]^J\colon [\cat{E},\cat{K}]^J\to [\cat{D},\cat{K}]^J
\]
that takes any object $\pi \in [\cat{E},\cat{K}]^J$ to 
$[q,\cat{K}]^J(\pi)=\pi\circ q$ and oplax transformation $f=\{f_E\colon \pi(E)\to \pi'(E)\}_{E\in \cat{E}}$ between $\pi,\pi' \in  [\cat{E},\cat{K}]^J$ to $[q,\cat{K}]^J(f)_D = f_{q(D)}$.

To keep the paper more succinct, we only focus on $\cat{D}$  being a monoid category (\ie a one-object category). Let $M=(M,\cdot,1)$ be a monoid and write $\cat{M}$ for the corresponding category.  
In this scenario, any lax functor $\pi \in [\cat{M},\cat{K}]^J$ is equivalent to a family $\{\pi_m\colon \pi(\ast)\to \pi(\ast)\}_{m\in M}$ of endomorphisms in $\cat{K}$ over a common object $\pi(\ast)$ and with the following property:
\begin{equation*}
 id_{\pi(\ast)} \leq \pi_1 \qquad \text{and} \qquad \pi_{m} \circ \pi_{n}\leq \pi_{m\cdot n}\text{.} 
\end{equation*}
An oplax transformation from $\pi$ to $\pi'$, both lax functors in $[\cat{M},\cat{K}]^J$, is an arrow $f$ in $\cat{J}$ such that the diagram below holds for any $m \in M$.
\begin{equation*}
	\begin{tikzpicture}[diagram,xscale=1.2]
		\node (n0) at (0,1) {$\pi (\ast)$};
		\node (n1) at (1,1) {$\pi'(\ast)$};
		\node (n2) at (0,0) {$\pi (\ast)$};
		\node (n3) at (1,0) {$\pi'(\ast)$};
		\draw[->] (n0) to node {$J(f)$} (n1);
		\draw[->] (n0) to node[swap] {$\pi(m)$} (n2);
		\draw[->] (n1) to node {$\pi'(m)$} (n3);
		\draw[->] (n2) to node[swap] {$J(f)$} (n3);
		\node[rotate=45] at ($(n0)!.5!(n3)$) {$\leq$};
	\end{tikzpicture}
\end{equation*}
The curious reader is referred to \cite{brengos2015:corr} and \cref{sec:timed-equivalences} for examples of these categories.

\subsection{Saturation}
\label{sec:general-saturation}

Building on \citeauthor{sobocinski:jcss}'s work on relational presheaves and weak bisimulation for LTSs, \cite{sobocinski:jcss}, \citeauthor{brengos2015:corr} rediscovered (special) saturation in the setting of lax functors in terms of the change-of-base functor induced by final monoid homomorphisms (seen as functors between index categories) \cite{brengos2015:corr}. This section pushes this approach further by considering arbitrary monoid homomorphisms. The main advantage of this generalisation is that different homomorphisms result in different aspects being abstracted by the saturation process (\eg internal moves, time durations).

In the sequel, let $q\colon M \to N$ be any monoid homomorphism and regard it as a functor between the associated categories.

\begin{definition}
	An order enriched category $\cat{K}$ is said to admit \emph{$q$-saturation} (\wrt a wide subcategory $J\colon \cat{J} \to \cat{K}$) whenever the change-of-base functor $[q,\cat{K}]^J\colon [\cat{N},\cat{K}]^J\to [\cat{M},\cat{K}]^J$ admits a left strict $2$-adjoint (denoted as $\Sigma_q^{J}$ or simply as $\Sigma_q$).
	\begin{equation}
		\begin{tikzpicture}[diagram]
			\node (n0) {$[\cat{M},\cat{K}]^J$};
			\node[right=6ex of n0] (n1) {$[\cat{N},\cat{K}]^J$};
			\draw[->,bend left] (n0.north east) to node {$\Sigma_q^J$} (n1.north west);
			\draw[->,bend left] (n1.south west) to node{$[q,\cat{K}]^J$} (n0.south east);
			\node[rotate=-90] at ($(n0.east)!.5!(n1.west)$) {$\dashv$};
		\end{tikzpicture}
		\label{eq:general-saturation-adjunction}
	\end{equation}
\end{definition}

This is indeed a generalisation of special saturation from \cite{brengos2015:lmcs,brengos2015:jlamp}: if we take as $q$ the unique morphism from the monoid $(\mathbb{N},+,0)$ of natural numbers under addition into the trivial one object monoid then, 
$[\mathbb{N},\cat{K}]^J \cong \cat{End}_J^{\leq}(\cat{K})$, 
$[\cat{1},\cat{K}]^J \cong \cat{End}_J^{\ast\leq}(\cat{K})$, 
and \eqref{eq:special-saturation-adjunction} becomes \eqref{eq:general-saturation-adjunction}.

In order to prove $q$-saturation admittance below, we work with a stronger type of order enrichment on $\cat{K}$\footnote{%
	In the conference version of the paper \cite{brengos2016:concur} we assumed the category $\cat{K}$ to admit finite joins, be $\mathbf{DCpo}$-enriched (\ie each hom-set is a complete order with directed suprema being preserved by the composition) and left distributive on finite joins. However, any such category admits arbitrary non-empty suprema and preserves them by the composition on the left. Hence, we decided to simplify our assumptions. All Kleisli categories considered in \cite{brengos2016:concur} and this paper satisfy these assumptions.
}. To be more precise we assume that $\cat{K}$ is left distributive order enriched category with arbitrary non-empty joins. 
Under these assumption $\cat{K}$ admits $q$-saturation \wrt $J$ whenever $q$ is a surjective monoid homomorphism \ie whenever the corresponding functor is full. The following theorem provides a direct construction for $\Sigma_q$.
\begin{theorem}
	\label{thm:general-saturation}
	If $q\colon \cat{M}\to \cat{N}$ is a full functor and $\cat{K}$ a left distributive category with arbitrary non-empty joins then, $\cat{K}$ admits $q$-saturation with respect to any $J$.
	The functor $\Sigma_q\colon [\cat{M},\cat{K}]^J\to [\cat{N},\cat{K}]^J$ acts as the identity on morphisms and takes any $\pi\in [\cat{M},\cat{K}]^J$ to the lax functor given on the only object as
	\begin{align*}
		\Sigma_q(\pi)(\ast) = {} & \pi(\ast)\\
	\intertext{and on any morphism $n \in N$ as}
		\Sigma_q(\pi)_n = {} & \bigvee_{i < \omega} \Pi_{n,i}\\
	\intertext{where $\Pi_{n,i}$ is inductively defined as}
		\Pi_{n,0} = {} & \bigvee \{\pi_m \mid q(m) = n\}\\
		\Pi_{n,i+1} = & {} \bigvee \left\{\Pi_{n_1,i}\circ \ldots\circ  \Pi_{n_l,i}\,\middle|\,n_1\cdot\dots\cdot n_l = n\right\}\text{.}
	\end{align*}
\end{theorem}

\begin{proof}
	At first we show that $\Sigma_q(\pi)$ is a well defined lax functor in $[\cat{N},\cat{K}]^J$. Note that the family $\Pi_n = \{\Pi_{n,i}\}_{i < \omega}$ is an ascending family of sets, that $\Pi_n$ is, by surjectivity of $q$, a  non-empty  directed set, and that $\Pi_n\circ \Pi_{n'}\subseteq \Pi_{n\cdot n'}$ for any $n,n\in N$. As a consequence, it holds that:
	\begin{gather*}
		id \leq \bigvee \Pi_1 \leq \Sigma_q(\pi)_1
		\\
		\Sigma_q(\pi)_{n} \circ \Sigma_q(\pi)_{n'} = 
		\bigvee \Pi_{n} \circ \bigvee \Pi_{n'} \stackrel{(\dagger)}{=} 
		\bigvee \Pi_{n}\circ \Pi_{n'} \stackrel{(\ddagger)}{\leq} 
		\bigvee \Pi_{n\cdot n'} = 
		\Sigma_{q}(\pi)_{n\cdot n'}
	\end{gather*}
	where $(\dagger)$ follows from left distributivity of $\cat{K}$ and 
	$(\ddagger)$ from definition of lax functor and construction of $\Pi_{n}$,$\Pi_{n'}$,$\Pi_{n\cdot n'}$.
	This proves the first statement. Now, we will verify that any oplax transformation in $[\cat{M},\cat{K}]^J$ is mapped onto an oplax transformation in $[\cat{N},\cat{K}]^J$. 
	Let $f\colon \pi \to \pi'$ be a morphism in $[\cat{M},\cat{K}]^J$ or, equivalently, a morphism $f$ in $\cat{J}$ with the property that $J(f)\circ \pi_m  \leq  \pi'_m\circ J(f)$ for any $m$ in $\cat{M}$.
	It follows from $J$-left distributivity and Scott continuity of the composition that:
	\begin{equation*}
	J(f)\circ \Sigma_{q}(\pi)_n = J(f)\circ \bigvee \Pi_n = \bigvee J(f)\circ \Pi_n \leq \bigvee \Pi'_n \circ J(f)  = \Sigma_{q}(\pi')_n \circ J(f)
		\text{.}
	\end{equation*}
	To complete the proof we have to show that for any lax functor $\pi\in [\cat{M},\cat{K}]^J$ and $\pi'\in [\cat{N},\cat{K}]^J$ the poset $[\cat{M},\cat{K}]^J(\pi, [q,\cat{K}](\pi'))$ is isomorphic to $[\cat{N},\cat{K}]^J(\Sigma_q(\pi),\pi')$.  Indeed:
	\begin{align*}
		& 
		\forall m \in M\  J(f)\circ \pi_m \leq \pi'_{q(m)}\circ J(f) \stackrel{(\dagger)}{\iff}
		\\& 
		\forall m,m_1\ldots m_k\in M\  J(f)\circ(\pi_{m_1} \vee \ldots \vee \pi_{m_k}) \leq  \pi'_{q(m)}\circ J(f) \land \forall i \  q(m_i)=q(m) \iff 
		\\&
		\forall m\in M\  J(f)\circ \bigvee \Pi_{q(m),0} \leq \pi'_{q(m)}\circ J(f) \stackrel{(\ddagger)}{\iff}
		\\&
		\forall m\in M\  \forall n < \omega, J(f)\circ \bigvee \Pi_{q(m),n} \leq \pi'_{q(m)}\circ J(f) \iff 
		\\& 
		\forall m \in M\  J(f)\circ \bigvee \Pi_{q(m)} \leq \pi'_{q(m)}\circ J(f)\iff 
		\forall m\in M\  J(f)\circ \Sigma_q(\pi)_{q(m)} \leq \pi'_{q(m)}\circ J(f)
		\text{.}
	\end{align*}
	The equivalence $(\dagger)$ follows by $J$-left distributivity of $\cat{K}$ and $(\ddagger)$ by induction on $n < \omega$.
\end{proof}

If $\cat{K}$ is also right distributive (hence $J$-right distributive for any $J$) then, the formula for $\Sigma_q$ simplifies considerably as stated by \cref{thm:simplified-formula-for-sigma} below.

\begin{theorem}
	\label{thm:simplified-formula-for-sigma}
	If the hom-posets of $\cat{K}$ admit arbitrary non-empty suprema and composition preserves them then, for any $\pi\in [\cat{M},\cat{K}]^J$ it holds that:
	\begin{equation*}
		\Sigma_q(\pi)_n = \bigvee\left\{ \pi_m \,\middle|\, q(m) = n \right\}
		\text{.}
	\end{equation*}
\end{theorem}
\begin{proof}
	Observe that $\Pi_n = \Pi_{n,0}$ and hence $\Sigma_q(\pi)_n  = \bigvee \Pi_n  = \bigvee \Pi_{n,0} = \bigvee_{m: q(m) = n} \pi_m$. 
\end{proof}

Saturation is idempotent. The claim readily follows from the fact that the underlying 2-adjunction is a coreflection and splits the identity on the subcategory.
\begin{proposition}
	\label{thm:composition-comonad}
	The functor $\Sigma_q\circ [q,\cat{K}]^J$ is the identity on $[\cat{N},\cat{K}]^J$.
\end{proposition}

\begin{proof}
It is enough to show that for any $\pi'\in [\cat{N},\cat{K}]^J$ we have $\Sigma_q(\pi'\circ q) = \pi'$. Take $\pi = \pi'\circ q$ and note that $
\Pi_{w,0} = \{\pi_{m_1} \vee \ldots \vee \pi_{m_k} \mid q(m_i) = w\} = \{\pi'_{w} \vee \ldots \vee \pi'_{w} \mid q(m_i) = w\} =\{\pi'_w\}$ for $w\in N$.
Moreover, since $\pi'$ is a lax functor we can easily prove by induction that $t\leq \pi'_w$ holds for any $t\in \Pi_{w,n}$. Hence, $\Sigma_q(\pi'\circ q)_w= \bigvee \Pi_w = \pi'_w$ which proves the assertion.
\end{proof}

Saturation is compositional in the sense that $q_2 \circ q_1$-saturation can be staged in terms of $q_1$-saturation and $q_2$-saturation.
\begin{proposition}
	\label{thm:general-saturation-compsition}
	The functor $\Sigma_{q_2}\circ \Sigma_{q_1}$ is (naturally isomorphic to) $\Sigma_{q_2\circ q_1}$.
\end{proposition}

\begin{proof}
	This statement follows directly from the fact that $\Sigma_{q_1}$ and $\Sigma_{q_2}$ are left adjoints and that $[q_1,\cat{K}]^J\circ [q_2,\cat{K}]^J = [q_2\circ q_1,\cat{K}]^J$.
\end{proof}

Natural transformations are indeed oplax transformations hence a special class of morphisms between lax functors. This class of morphisms is preserved by saturation whenever composition with morphisms from $J$ preserves non-empty joins.

\begin{proposition}
	\label{thm:strong-natural}
	Assume that $\cat{K}$ is $J$-left and $J$-right distributive. For any $f\colon \pi \to \pi'$ in $[\cat{N},\cat{K}]^J$, if $f$ is a natural transformation then, so is $f\colon \Sigma_q(\pi) \to \Sigma_q(\pi')$.
\end{proposition}

\begin{proof}
	Recall that to be (ordinary) natural transformation means that 
	$f\circ \pi_m  = \pi'_m \circ f$ and $f\circ \Sigma_q(\pi)_n  = \Sigma_q(\pi')_n \circ f$ for any $m \in M$ and $n\in N$.
	Observe that:
	\begin{equation*}
		f\circ (\pi_{m_1}\vee \ldots \vee \pi_{m_k} ) \stackrel{J\text{-LD}}{=} f\circ \pi_{m_1}\vee \ldots \vee f\circ \pi_{m_k}  =	\pi'_{m_1}\circ f\vee \ldots \vee \pi'_{m_k} \circ f \stackrel{J\text{-RD}}{=}  (\pi'_{m_1}\vee \ldots \vee \pi'_{m_k} )\circ f
		\text{.}
	\end{equation*}	
	Hence, for any $n \in N$ it holds that $J(f)\circ \Pi_{n,0} = \Pi'_{n,0}\circ J(f)$ and, by induction on $i$, $J(f)\circ \Pi_{n,i} = \Pi_{n,i}'\circ J(f)$. We conclude that $J(f)\circ \Pi_{n} = \Pi_{n}'\circ J(f)$ which proves the thesis.
\end{proof}

Let $I\colon \cat{I} \to \cat{J}$ exhibit $\cat{I}$ as a wide subcategory of $\cat{J}$ and observe that hom-posets of $[\cat{M},\cat{K}]^{J\circ I}$ are included in those of $[\cat{M},\cat{K}]^{J}$. As a consequence, there is an obvious inclusion 2-functor $\Upsilon^{I}_{\cat{M}}\colon [\cat{M},\cat{K}]^{J\circ I} \to [\cat{M},\cat{K}]^J$. 
Saturation commutes with the ``change of hom-posets'' induced by $I$.

\begin{proposition}
	\label{thm:change-of-subcat}
	The following is a morphism of strict 2-adjunctions.
	\begin{equation*}
		\begin{tikzpicture}[diagram]
			\node                   (n0) {$[\cat{N},\cat{K}]^{J\circ I}$};
			\node[below= 6ex of n0] (n1) {$[\cat{M},\cat{K}]^{J\circ I}$};
			\node[right=13ex of n0] (n2) {$[\cat{N},\cat{K}]^{J}$};
			\node[below= 6ex of n2] (n3) {$[\cat{M},\cat{K}]^{J}$};
			
			\def\d{1ex}
			\def\s{-1.0ex}
			
			\draw[->,bend left] ([xshift=\s+\d]n0.south) to node (l0) {$[q,\cat{K}]^{J\circ I}$} ([xshift=\s+\d]n1.north);
			\draw[->,bend left] ([xshift=\s-\d]n1.north) to node (l1) {$\Sigma_q^{J\circ I}$} ([xshift=\s-\d]n0.south);
			\node[] at ($(l0.west)!.5!(l1.east)$) {$\dashv$};

			\def\d{1ex}
			\def\s{-.3ex}
			
			\draw[->,bend left] ([xshift=\s+\d]n2.south) to node (l2) {$[q,\cat{K}]^J$} ([xshift=\s+\d]n3.north);
			\draw[->,bend left] ([xshift=\s-\d]n3.north) to node (l3) {$\Sigma_q^J$} ([xshift=\s-\d]n2.south);
			\node[] at ($(l2.west)!.5!(l3.east)$) {$\dashv$};
			
			\draw[->] (n0) to node       (l4) {$\Upsilon^{I}_{\cat{N}}$} (n2);
			\draw[->] (n1) to node[swap] (l5) {$\Upsilon^{I}_{\cat{M}}$} (n3);

				\node[] at ($(l4.south)!.5!(l5.north)$) {$\cong$};
		\end{tikzpicture}
	\end{equation*}
\end{proposition}

\begin{proof}
	Observe that $[q,\cat{K}]^{J\circ I}$ is the restriction of $[q,\cat{K}]^{J}$ to the wide subcategories identified by $\Upsilon^{I}_{\cat{M}}$ and $\Upsilon^{I}_{\cat{N}}$ and recall from \cref{thm:composition-comonad} that both adjunctions are coreflections.
\end{proof}

\begin{theorem}
	If $\cat{K}$ admits $q$-saturation with respect to $J$ then, it admits $q$-saturation with respect to $J \circ I$.
\end{theorem}

\begin{proof}
	Follows directly from \cref{thm:general-saturation,thm:change-of-subcat}.
\end{proof}

\subsection{Behavioural equivalences}
\label{sec:general-saturation-equivalences}

In this section we extend the theory of saturation-based behavioural equivalences from special to general saturation.

Recall from \cref{sec:special-saturation-equivalences} that saturation-based equivalences are defined as kernel pairs of \emph{behavioural morphisms} \ie endomorphism whose domain is the saturation of the system under scrutiny. 
\Cref{def:q-behavioural-equivalence} below is a generalisation of this notion where special saturation is replaced with general saturation. As in the previous section, we assume that $\cat{K}$ is a left distributive, order enriched category with arbitrary non-empty joins, that $J$ exhibits a wide subcategory $\cat{J}$ of $\cat{K}$, and that $q\colon M \to N$ is a surjective homomorphism of monoids.

\begin{definition}
	\label{def:q-behavioural-equivalence}
	A \emph{$q$-behavioural morphism} for $\pi\in \JFun{\cat{M}}{\cat{K}}{J}$ is any morphism $f$ from $J$ that carries a (strict) natural transformation with domain $\Sigma_q^J(\pi) \in \JFun{\cat{N}}{\cat{K}}{J}$ \ie any $f$ with the property that there is $\pi'\in \JFun{\cat{N}}{\cat{K}}{J}$ such that the diagram below commutes for any $n \in \cat{N}$:
  \begin{equation*}
  	\begin{tikzpicture}[diagram,xscale=1.4,yscale=1.1]
  		\node (n0) at (0,1) {$\pi (\ast)$};
  		\node (n1) at (1,1) {$\pi'(\ast)$};
  		\node (n2) at (0,0) {$\pi (\ast)$};
  		\node (n3) at (1,0) {$\pi'(\ast)$};
  		\draw[->] (n0) to node {$J(f)$} (n1);
  		\draw[->] (n0) to node[swap] {$\Sigma_q^J(\pi)_n$} (n2);
  		\draw[->] (n1) to node {$\pi'_n$} (n3);
  		\draw[->] (n2) to node[swap] {$J(f)$} (n3);
  	\end{tikzpicture}
  \end{equation*}
	A \emph{$q$-bisimulation} for $\pi$ is a relation $R\rightrightarrows \pi(\ast)$ from $J$ that is also the kernel pair of some $q$-behavioural morphism for $\pi$.
\end{definition}

If $\cat{K}$ is the category of $T$-coalgebras then, the notion of $id$-bisimulation coincides with that of kernel bisimulation for $T$-coalgebras. We adopt the term ``refinement system'', used to denote to the codomain of the behavioural morphism associated with a kernel bisimulation, to denote the codomain ($\pi'$ in \cref{def:q-behavioural-equivalence}) of a $q$-behavioural morphism.

It follows from the idempotency property of saturation that $q$-bisimulation can be equivalently stated in terms of refinement systems from the original category (\ie from $\JFun{\cat{M}}{\cat{K}}{J}$ instead of $\JFun{\cat{N}}{\cat{K}}{J}$): it suffices to take their saturation.
\begin{proposition}
	A morphism $f$ from $J$ is a $q$-behavioural morphism for $\pi \in [\cat{M},\cat{K}]^J$ if and only if there is $\pi'\in [\cat{M},\cat{K}]^J$ such that the diagram below commutes for any $m \in \cat{M}'$.
    \begin{equation*}
    	\begin{tikzpicture}[diagram,xscale=1.4,yscale=1.1]
    		\node (n0) at (0,1) {$\pi (\ast)$};
    		\node (n1) at (1,1) {$\pi'(\ast)$};
    		\node (n2) at (0,0) {$\pi (\ast)$};
    		\node (n3) at (1,0) {$\pi'(\ast)$};
    		\draw[->] (n0) to node {$J(f)$} (n1);
    		\draw[->] (n0) to node[swap] {$\Sigma_q^J(\pi)_m$} (n2);
    		\draw[->] (n1) to node {$\Sigma_q(\pi')_m$} (n3);
    		\draw[->] (n2) to node[swap] {$J(f)$} (n3);
    	\end{tikzpicture}
    \end{equation*}
\end{proposition}

\begin{proof}
	By \cref{thm:composition-comonad}.
\end{proof}

Each notion of $q$-bisimulation for systems modelled in $[\cat{M},\cat{K}]^J$ arises from some congruence for $M$ and, \viceversa, each congruence defines a notion of $q$-bisimulation. \Cref{thm:behavioural-compostion,thm:behavioural-composition-bisimulation} below state that that coarser congruences define coarser notions of bisimulations. 
A prototypical example are strong/weak equivalences and time/time-abstract ones (see \cref{sec:timed-equivalences} and \cref{fig:morphisms-equivalence-spectrum}).

\begin{theorem}
	\label{thm:behavioural-compostion}
	Every $q$-behavioural morphism for~$\pi$ is a $(q' \circ q)$-behavioural morphism for~$\pi$.
\end{theorem}

\begin{proof}
	Follows directly from \cref{thm:general-saturation-compsition,thm:strong-natural}.
\end{proof}

\begin{corollary}
	\label{thm:behavioural-composition-bisimulation}
	Every $q$-bisimulation for $\pi$ is a $(q'\circ q)$-bisimulation for $\pi$.	
\end{corollary}

\Cref{thm:behavioural-inclusion,thm:behavioural-inclusion-bisimulation} below state that behavioural equivalences are preserved under ``change of hom-sets''. A prototypical instance is offered by bisimulations and language equivalences: since the former will be captured using the canonical inclusion $(-)^\sharp\colon \Set \to \kl(\ENA)$ and the latter using the canonical inclusion $(-)^\sharp\colon\kl(\mathcal{P}) \to \kl(\ENA)$, it follows from \cref{thm:behavioural-inclusion-bisimulation} that bisimulations are always language equivalences. 
\begin{theorem}
	\label{thm:behavioural-inclusion}
	Every $q$-behavioural morphism for $\pi \in \JFun{\cat{M}}{\cat{K}}{J\circ I}$
  is a $q$-behavioural morphism for $\pi \in \JFun{\cat{M}}{\cat{K}}{J}$.
\end{theorem}

\begin{proof}
	Follows directly from \cref{thm:change-of-subcat,thm:strong-natural}.
\end{proof}

\begin{corollary}
	\label{thm:behavioural-inclusion-bisimulation}
	A $q$-bisimulation for $\pi \in \JFun{\cat{M}}{\cat{K}}{J\circ I}$ is a $q$-bisimulation for $\pi \in \JFun{\cat{M}}{\cat{K}}{J}$.	
\end{corollary}

\section{Functor models of timed transition systems}
\label{sec:systems-as-functors}

To develop a behavioural theory for timed transition systems based on the framework we introduced in \cref{sec:general-saturation-and-equivalences} one only needs to model TTSs as (lax) functors where the index monoid category ($\cat{M}$) models abstract executions (\eg sequences of timed steps) and the base category ($\cat{K}$) models the effects associated with each abstract execution (\eg the set of states reachable with a certain sequence of timed steps). Then one can readily instantiate our framework as we illustrate in \cref{sec:timed-equivalences}.

Recall from \cref{sec:timed-transition-systems} that timed transition system are essentially transition systems whose transitions are labelled with symbols from $\Sigma_\tau$ and time durations from $\mathbb{T}$ (with some definitions imposing additional constraints).
There is bijective correspondence between transition systems with labels from
$\Sigma_\tau \times \mathbb{T}$ and time-indexed families of transition systems with labels $\Sigma_\tau$ and the same state space:
\begin{equation}
	\label{eq:tts-t-families-correspondence}
	\frac{
		X \to \mathcal{P}(\Sigma_\tau \times \mathbb{T} \times X)
	}{
		\mathbb{T} \to \mathcal{P}(\Sigma_\tau  \times X)^X
	}
\end{equation}
Observe that the first are coalgebras of type $\mathcal{P}(\Sigma_\tau \times \mathbb{T} \times -)$ and that the second are $\mathbb{T}$-indexed families of endomorphisms from $\kl\big(\LTS\big)$ since $\kl(\LTS)(X,X) = \mathcal{P}(\Sigma_\tau  \times X)^X$.
Let $\cat{T}^*$ be the category associated to the free monoid $\mathbb{T}^*$ and $J$ the inclusion $(-)^\sharp\colon\Set \to \kl(\LTS)$.
The correspondence extends to an isomorphism
\begin{equation*}
	\cat{Coalg}(\mathcal{P}(\Sigma_\tau \times \mathbb{T} \times -))
	\cong
	\Cat\big(\cat{T}^*,\kl\big(\LTS\big)\big){}^{J}
\end{equation*}
between the categories of $\mathcal{P}(\Sigma_\tau \times \mathbb{T} \times -)$-coalgebras and the subcategory of $\big[\cat{T}^*,\kl\big(\LTS\big)\big]{}^{J}$ given by (strict) functors and (strict) natural transformations.
Thus we have that:
\begin{equation*}
	\cat{Coalg}(\mathcal{P}(\Sigma_\tau \times \mathbb{T} \times -))
	\cong
	\Cat\big(\cat{T}^*,\kl\big(\LTS\big)\big){}^{J}
	\hookrightarrow
	\big[\cat{T}^*,\kl\big(\LTS\big)\big]{}^{J}
	\text{.}
\end{equation*}
We model TTSs as ordinary (\ie not lax) functors in $\big[\cat{T}^*,\kl\big(\LTS\big)\big]{}^{J}$ (or $\big[\cat{T}^*,\kl\big(\ENA\big)\big]{}^{J}$ when discussing language equivalences) and write $\underline{\alpha}$ and $\overline{\alpha}$ for the functor model and the coalgebra model of a TTS $\alpha$ whenever the distinction is not clear from the context.

Besides fitting into our framework, modelling timed transition systems as functors allow us to extend our approach to timed behaviours that do not have a direct coalgebra model (or an equivalent of \eqref{eq:tts-t-families-correspondence}) as it happens for timed Segala automata \cite{segala:tcs2002}. In fact, a single entry of a timed transition table can easily yield uncountably many transitions in the associated TTS whereas the semantics of Segala systems based on compound steps (also called \emph{convex semantics} after the use of convex closures, \cf \cite{segala:phd-thesis,varaccawinskel2006:mscs}) assumes that probability distribution supports are (finitely) bounded (\cf \cite{brengos2014:cmcs,jacobs08:cmcs}).
Instead, we are able to derive the behavioural theory of timed Segala systems using the standard convex semantics without modifications (\cf \cref{sec:s-ta}).

A general construction for defining functor models for timed calculi and timed automata is out of the scope of this work. We refer the interested reader to \cite{kick:phdthesis} for a coalgebraic account of timed calculi and to \cite{brengos2016:concur} an account of timed automata in terms for lax functors.

\section{Behavioural equivalences for timed transition systems}
\label{sec:timed-equivalences}

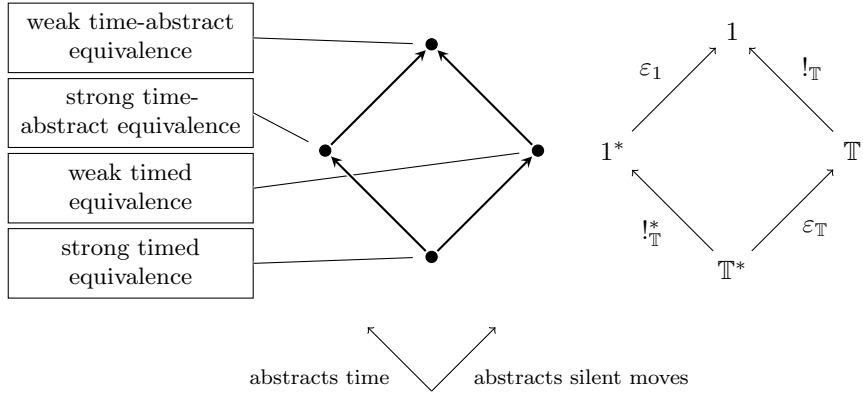
\begin{figure}[t]
	\centering
	\begin{tikzpicture}[
			auto, font=\footnotesize,
			scale=1,		
			box/.style={text width=#1, align=center,draw,fill=white},
			dot/.style={circle,draw,fill,inner sep=1.5pt,outer sep=1pt},
			box con/.style={-,thin,shorten >=4pt,crossing=1.7pt},
			hasse con/.style={-stealth,thick,crossing=2pt}
		]

		\node[box={3cm}] (b00) at (-4,-1.5) {strong timed equivalence};
		\node[box={3cm}] (b10) at (-4,-0.5) {weak timed equivalence};
		\node[box={3cm}] (b01) at (-4, 0.5) {strong time-abstract equivalence};
		\node[box={3cm}] (b11) at (-4, 1.5) {weak time-abstract equivalence};

		\begin{scope}[rotate=45]
		\foreach \x/\xc in {0/-1,1/1}{
			\foreach \y/\yc in {0/-1,1/1}{
				\node[dot] (p\x\y) at (\xc,\yc) {};
		}}
		\end{scope}

		\begin{scope}[
				yshift=-3.2cm,
			]
		\begin{scope}[
				rotate=45,
				scale=1.2,
				font=\scriptsize,
				-stealth,
				auto
			]
			\draw[->] (0,0) 
				to node [swap] 
				{abstracts silent moves}
				(1,0);
			\draw[->] (0,0) 
				to node 
				{abstracts time}
				(0,1);
		\end{scope}
		\end{scope}

		\foreach \x in {0,1}{
			\foreach \y in {0,1}{
			\draw[box con] (b\x\y.east) -- (p\x\y);
		}}
				
		\draw[hasse con] (p00) to (p10);
		\draw[hasse con] (p10) to (p11);
		\draw[hasse con] (p00) to (p01);
		\draw[hasse con] (p01) to (p11);

		\begin{scope}[xshift=4cm,rotate=45,diagram,scale=.75]
		 \node[] (n0) at (-1,-1) {$\mathbb{T}^\ast$};
		 \node[] (n1) at ( 1,-1) {$\mathbb{T}$};
		 \node[] (n2) at (-1, 1) {$1^\ast$};
		 \node[] (n3) at ( 1, 1) {$1$};
		 \draw[->] (n0) to node[swap] {$\varepsilon_{\mathbb{T}}$} (n1);
		 \draw[->] (n2) to node       {$\varepsilon_{1}$} (n3);
		 \draw[->] (n0) to node       {$!_{\mathbb{T}}^\ast$} (n2);
		 \draw[->] (n1) to node[swap] {$!_{\mathbb{T}}$} (n3);
		\end{scope}

	\end{tikzpicture}
	\caption{Monoid morphisms (right) and the corresponding spectrum of saturation-based timed behavioural equivalences (left).}
	\label{fig:morphisms-equivalence-spectrum}
\end{figure}

In this section we apply the general framework developed in \cref{sec:general-saturation-and-equivalences} to timed transition systems systems and show that equivalences of interests are all instances of the general notion of $q$-bisimulation; 
In particular, we rediscover the eight combinations of:
\begin{itemize} 
	\item timed or time-abstract, 
	\item strong or weak, and 
	\item bisimulation or (finite) language equivalence
\end{itemize}
Each pair in the list above corresponds to one of three orthogonal abstraction dimensions: 
\begin{itemize}
	\item
	abstraction over time arises from the unique homomorphism $!_{\mathbb{T}}$ going from the monoid $\mathbb{T}$ modelling time (\cf \cref{sec:timed-transition-systems}) to the trivial monoid $1$; 
	\item
	abstraction over unobservable moves arises from components of the counit $\varepsilon$ of the free monoid adjunction;
	\item
	abstraction over branching arises from the fact that the inclusion $\Set \to \kl(\ENA)$ factors through $\kl(\mathcal{P}) \to \kl(\ENA)$ by construction of $\ENA$.
\end{itemize}
The first two dimensions define four types of equivalence that, as a consequence of \cref{thm:behavioural-composition-bisimulation}, are organised with respect to their discriminating power in the spectrum depicted in \cref{fig:morphisms-equivalence-spectrum} together with the corresponding monoid homomorphisms.
These equivalences are always organised in this spectrum independently from any specific choice made for the subcategory $J$ sourcing behavioural morphisms. This means that spectra for bisimulations and their language equivalent counterpart are alike. 
The third dimension defines a bisimulation equivalence and a language equivalence counterpart for each combination of abstraction over the first and second aspect where, as a consequence of \cref{thm:behavioural-inclusion-bisimulation}, the former is always more discriminating that the later.
It follows that the above spectra for bisimulations and language equivalences are actually two opposing faces in the spectrum depicted in \cref{fig:timed-equivalences-spectrum}.

For exposition convenience, in the sequel we shall use the term $q$-language equivalence instead of $q$-bisimulation to signal that $\cat{K}$ and $\cat{J}$ are put to be $\kl(\ENA)$ and $\kl(\mathcal{P})$, and otherwise use $q$-bisimulation in the setting where $\cat{K}$ and $\cat{J}$ are put to be $\kl(\LTS)$ and $\Set$.

Although we focus on non-determinism, constructions described in this section apply to other computational effects modelled by suitable monads as we will discuss in the second part of this work.

\subsection{Timed bisimulation}
Recall from \cref{sec:systems-as-functors} that a ordinary functor $\pi$ over a
free monoid like $\mathbb{T}^\ast$ describes the transitions of a system
together with all their self-compositions. A $id_M$-behavioural morphism for $\pi$ is a morphism $f$ with the property that 
\[
	f\circ \pi_m = \pi'_m\circ f
\]
for some $\pi'$ and every $m \in M$ since $id$-saturation acts as the identity. When all components of $\pi$ and $\pi'$ are coalgebras a relation $R \rightrightarrows \pi(\ast)$ is a $id$-bisimulation for $\pi$ if and only if it is a (kernel) bisimulation for every component of $\pi$. 
Thus, we obtain the notion timed bisimulation for TTSs as that of $id_{\mathbb{T}^\ast}$-bisimulation for their associated functor model.

\begin{proposition}
	\label{thm:tts-q-timed-bisimulation}
	For a TTS $\alpha$, $\underline{\alpha}$ its functor model, and $R$ an equivalence relation on its carrier:
	$R$ is a timed bisimulation for ${\alpha}$ iff it is a $id_{\mathbb{T}^\ast}$-bisimulation for $\underline{\alpha}$.
\end{proposition}

\begin{proofatend}
	Let $\alpha$ be a TTS.
	Assume $R\subseteq X^2$ is a timed bisimulation for $\alpha$ and let $f\colon X \to  Y \cong X/R$ be the canonical projection in \Set.
	Let $\underline{\beta}\colon \mathbb{T}^\ast \to \cat{K}$ be the lax functor defined by the family $\{\beta_t\colon Y \to \mathcal{P}(\Sigma_\tau \times Y)\}_{t \in \mathbb{T}}$ where:
	\[
	  \beta_t(y) \defeq \{ (\sigma, y') \mid \forall x \in f^{-1}(y)\, \exists x' \in f^{-1}(y') \text{ s.t. } x \xrightarrow{(\sigma,t)} x' \}\text{.}
	\]
	From the definition of $\beta_t$ it is immediate to check that for any $t \in \mathbb{T}$:
	\[
	  (\sigma,f(x')) \in (\underline{\beta}_t \circ f)(x) \iff (\sigma,x') \in (\underline{\alpha}_t)(x)
	\]
	and hence that $f \circ \Sigma_{id_{\mathbb{T}^\ast}}(\underline{\alpha}) = \underline{\beta} \circ f $. Therefore, the quotient map $f$ induced by $R$ extends to a $id_{\mathbb{T}^\ast}$-behavioural morphism for $\underline{\alpha}$.
	
	For the converse implication take $R$ to be a $id_{\mathbb{T}^\ast}$-bisimulation for $\underline{\alpha}$. This means that there is a lax functor $\beta \in [\mathbb{T}^\ast,\cat{K}]$ whose carrier $Y$ is given by the set of abstract classes of $R$ and such that the quotient map $f\colon X\to Y$ induced by $R$ satisfies:
	\[
	  f \circ \Sigma_{id_{\mathbb{T}^\ast}}(\underline{\alpha}) = \underline{\beta} \circ f
	  \text{.}
	\]
	In particular, for any $(\sigma,t) \in \Sigma_\tau\times \mathbb{T}$ we have that $(\sigma,f(x')) \in (\underline{\beta}_t \circ f)(x) \iff (\sigma,x') \in (\underline{\alpha}_t)(x)$ from which it immediately follows that $R$ is a timed bisimulation for $\alpha$.
\end{proofatend}

\subsection{Abstraction over time}
Abstraction over time is captured via saturations that forget some of the time information attached to steps like the $!_{\mathbb{T}}^\ast\colon \mathbb{T}^\ast \to 1^\ast$ (which forgets durations entirely). In general, forgetting some of the information attached to steps while preserving the steps corresponds to $q^\ast\colon M^\ast \to N^\ast$ for some $q\colon M \to N$. The functor $\Sigma_{q^\ast}\colon [\cat{M}^\ast,\cat{K}]^J \to [\cat{N}^\ast,\cat{K}]^J$ is given on a objects as:
\begin{equation*}
  \Sigma_{q^\ast}(\pi)_{n_1\dots n_k} = 
  \bigvee_{n_i = q(m_i)} \pi_{m_1\dots m_k} =
  \bigvee_{n_i = q(m_i)} \pi_{m_1} \circ \dots \circ \pi_{m_k}
  \text{.}
\end{equation*}
Steps in the saturated system are combinations of all steps associated to a pre-image through $q$; in particular, $\Sigma_{q^\ast}(\pi)_{n} = \bigvee_{n = q(m)} \pi_{m}$.
If $q$ is $!_{\mathbb{T}}^\ast$, steps that differ only for their duration are combined allowing thus to rediscover time-abstract bisimulation as $!_{\mathbb{T}}^\ast$-bisimulation.

\begin{proposition}
	\label{thm:tts-q-time-abstract-bisimulation}
	For a TTS $\alpha$, $\underline{\alpha}$ its functor model, and $R$ an equivalence relation on its carrier:
	$R$ is a time-abstract bisimulation for ${\alpha}$ iff it is a $!_{\mathbb{T}}^\ast$-bisimulation for $\underline{\alpha}$.
\end{proposition}

\begin{proofatend}
	Let $\alpha$ be a TTS.
	Assume $R\subseteq X^2$ is a weak timed bisimulation for $\alpha$ and let $f\colon X \to  Y \cong X/R$ be the canonical projection in \Set.
	Let $\underline{\beta}\colon \mathbb{T} \to \cat{K}$ be the lax functor defined on each $t \in \mathbb{T}$ and $y \in Y$ as:
	\[
	   \beta_t(y) \defeq \bigcup_{t_1\dots t_n \in \varepsilon_{\mathbb{T}}^{-1}(t) }
	   \left\{
	      (\sigma_i, y')
	   \,\middle|\,
	   \begin{array}{l}
	      \forall x \in f^{-1}(y)\, \exists x' \in f^{-1}(y') \text{ such that } \\
	      x \xrightarrow{(\sigma_1,t_1)} \cdots \xrightarrow{(\sigma_i,t_i)} \cdots \xrightarrow{(\sigma_n,t_n)} x' \text{ and } j \neq i \implies \sigma_j = \tau
	   \end{array}
	   \right\}\text{.}
	\]
	From the definition of $\beta_t$ it is immediate to check that for any $t \in \mathbb{T}$ and $\sigma \in \Sigma_\tau$:
	\[
	   (\sigma,f(x')) \in (\underline{\beta}_t \circ f)(x) \iff (\sigma,x') \in \Sigma_{\varepsilon_{\mathbb{T}}}(\underline{\alpha})_t(x)
	\]
	and hence that the quotient map $f$ for $R$ extends to a $\varepsilon_{\mathbb{T}}$-behavioural morphism for $\underline{\alpha}$.
	
		For the converse implication take $R$ to be a $\varepsilon_{\mathbb{T}}$-bisimulation for $\underline{\alpha}$. This means that there is a lax functor $\beta \in [\mathbb{T},\cat{K}]$ whose carrier $Y$ is given by the set of abstract classes of $R$ and such that the quotient map $f\colon X\to Y$ induced by $R$ satisfies:
	\[
	   f \circ \Sigma_{\varepsilon_{\mathbb{T}}}(\underline{\alpha}) = \underline{\beta} \circ f
	   \text{.}
	\]
	In particular, for any $t \in \mathbb{T}$ and $\sigma \in \Sigma_\tau$ we have that $(\sigma,f(x')) \in (\underline{\beta}_t \circ f)(x)$ iff
	$(\sigma,x') \in \bigvee_{\vec{t} \in \varepsilon_{\mathbb{T}}^{-1}(t)} \underline{\alpha}_{\vec{t}}(x)$. By straightforward definition expansion we conclude that $R$ is a weak timed bisimulation for $\alpha$.
\end{proofatend}

\subsection{Abstraction over unobservable moves}

Abstraction over unobservable moves arises from components of the counit $\varepsilon\colon (-)^\ast \to \Id$ of the free monoid adjunction \ie  arrows $\varepsilon_M\colon M^\ast \to M$ taking $(m_1,\dots,m_n)$ to $m_1 \cdot \ldots \cdot m_n$ where $\cdot$ is the monoid operation. 
By \cref{thm:simplified-formula-for-sigma}, if $\cat{K}$ admits and preserves arbitrary non-empty joins the functor $\Sigma_{\varepsilon_{M}}\colon [\cat{M}^\ast,\cat{K}]^J \to [\cat{M},\cat{K}]^J$ is given every $\pi$ as follows:
\begin{equation*}
  \Sigma_{\varepsilon_{M}}(\pi)_m = 
  \bigvee_{\varepsilon_{M}(\vec{m}) = m} \pi_{ \vec{m}} = 
  \bigvee_{m_1 \cdot  {\dots} \cdot m_k = m} \pi_{m_1} \circ \dots \circ \pi_{m_k}\text{.}
\end{equation*}
Steps in the saturated system are combinations of all sequences in the original one associated with any decomposition of the stage $m$. In particular, when $\cat{K}$ models unobservable moves, $\varepsilon_M$-bisimulation can be seen as the weak counterpart of strong behavioural equivalence for systems modelled in $[\cat{M}^\ast,\cat{K}]^J$, \ie, $id_{M^\ast}$-bisimulation.

Let $M$ be a monoid $\mathbb{T}$ modelling time and let $q$ be $\varepsilon_{\mathbb{T}}$. The resulting notion of $q$ bisimulation saturates a sequence of timed steps into one whose duration is the total duration of the sequence thus abstracting from intermediate steps division only.
\begin{proposition}
	\label{thm:tts-q-weak-timed-bisimulation}	
	For a TTS $\alpha$, $\underline{\alpha}$ its functor model, and $R$ an equivalence relation on its carrier:
	$R$ is a weak timed bisimulation for ${\alpha}$ iff it is a $\varepsilon_{\mathbb{T}}$-bisimulation for $\underline{\alpha}$.
\end{proposition}

\begin{proofatend}
	Assume $R\subseteq X^2$ is a time-abstract bisimulation for $\alpha$ and let $f\colon X \to  Y \cong X/R$ be the canonical projection in \Set. Define $\beta\colon Y\to \mathcal{P}(\Sigma_\tau\times Y)$ by:
	\[
		\beta(y) = \{(\sigma,y')\mid \forall x\in y \exists x'\in y' \text{ and }\exists t \text{ s.t. }(\sigma,t,x')\in \alpha(x)  \}
		\text{.}
	\]
	Then the morphism $f\colon X\to Y$ in $\Set$ which maps any element to its abstract class satisfies $f\circ \Sigma_{\sharp}(\underline{\alpha})_1 = \beta \circ f$ (here, $\circ$ denotes the composition in $\kl(\mathcal{P}^\Sigma)$). Note that $\Sigma_{!_{\mathbb{T}}^\ast}(\alpha)_n = (\Sigma_{!_{\mathbb{T}}^\ast}(\alpha)_1)^n$ and put $\underline{\beta}\in [1^\ast,\cat{K}]$ to be $\underline{\beta}_n = \beta^n$. We have $f\circ \Sigma_{!_{\mathbb{T}}^\ast}(\underline{\alpha})_n = \underline{\beta}_n \circ f$. This proves that $R$ is a $!_{\mathbb{T}}^\ast$-bisimulation.
	
	For the converse implication take $R$ to be a $!_{\mathbb{T}}^\ast$-bisimulation for $\underline{\alpha}$. This means that there is a lax functor $\pi' \in [\mathbb{N},\cat{K}]$ whose carrier $Y$ is given by the set of abstract classes of $R$ for which the map $f\colon X\to Y$ which assigns any $x$ to its abstract class satisfies $
	f\circ \Sigma_{!_{\mathbb{T}}^\ast}(\underline{\alpha})_n =\pi'_n \circ f \text{ for } n\in \mathbb{N}$.
	In particular the above equality holds for $n=1$. By the construction of $\Sigma_{!_{\mathbb{T}}^\ast}(\underline{\alpha})_1$ it instantly follows that $R$ is a time-abstract bisimulation for $\alpha$.
\end{proofatend}

Let $M$ be the trivial monoid $1$ and observe that $1^\ast$ is isomorphic to the monoid $\mathbb{N}$ of natural numbers under addition. In this setting, system steps are essentially devoid of any time information, step sequences are associated only with their length, and the corresponding component of $\varepsilon$, $!_{\mathbb{N}}\colon \mathbb{N} \to 1$, simply forgets lengths.
This instance of $q$-bisimulation corresponds exactly to the construction used by \citeauthor{brengos2015:corr} in \cite{brengos2015:corr} to capture (untimed) weak bisimulation: in \loccit systems are modelled as lax functors over the monoid $\mathbb{N}$ of natural numbers under addition and saturation is defined in terms of a strict 2-adjunction to the category of lax functors over the trivial monoid $1$. 

We can combine abstraction over time and abstraction over unobservable moves we simply by composing their defining monoid morphisms as $\varepsilon_N \circ q^\ast$ or as $q \circ \varepsilon_M$; by naturality of $\varepsilon$ the notions of saturation coincide:
\begin{equation*}
  \Sigma_{\varepsilon_N \circ q^\ast}(\pi)_{n} = 
  \Sigma_{q \circ \varepsilon_M }(\pi)_{n} =
  \bigvee_{\substack{n_1 \cdot {\dots} \cdot n_k = n \\ n_i = q(m_i)}} \pi_{m_1} \circ \dots \circ \pi_{m_k} =
  \bigvee_{\substack{n = q(m)\\m_1 \cdot {\dots} \cdot m_k = m}} \pi_{m_1} \circ \dots \circ \pi_{m_k}
  \text{.}
\end{equation*}
In particular, when modelling timed transition systems, we have that
\[
	\begin{tikzpicture}[diagram,rotate=45,scale=.5]
		\node[] (n0) at (-1,-1) {$\mathbb{T}^\ast$};
		\node[] (n1) at ( 1,-1) {$\mathbb{T}$};
		\node[] (n2) at (-1, 1) {$1^\ast$};
		\node[] (n3) at ( 1, 1) {$1$};
		\draw[->] (n0) to node[swap] {$\varepsilon_{\mathbb{T}}$} (n1);
		\draw[->] (n2) to node       {$\varepsilon_{1}$} (n3);
		\draw[->] (n0) to node       {$!_{\mathbb{T}}^\ast$} (n2);
		\draw[->] (n1) to node[swap] {$!_{\mathbb{T}}$} (n3);
		\draw[->] (n0) to node       {$!_{\mathbb{T}^\ast}$} (n3);
	\end{tikzpicture}
\]
and we rediscover weak time-abstract bisimulation as $!_{\mathbb{T}^\ast}$-bisimulation.
\begin{proposition}
	\label{thm:tts-q-weak-time-abstract-bisimulation}
	For a TTS $\alpha$, $\underline{\alpha}$ its functor model, and $R$ an equivalence relation on its carrier:
	$R$ is a weak time-abstract bisimulation for ${\alpha}$ iff it is a $!_{\mathbb{T}^\ast}$-bisimulation for $\underline{\alpha}$.
\end{proposition}

\begin{proofatend}
	Assume $R\subseteq X^2$ is a weak time-abstract bisimulation for $\alpha$ and let $f\colon X \to  Y \cong X/R$ be the canonical projection in \Set.
	Let $\underline{\beta} \in \cat{K}$ be the LTS given on each $y \in Y$ as:
	\[
	  \beta(y) \defeq \bigcup_{t_1\dots t_n \in \mathbb{T}^\ast}
	  \left\{
	     (\sigma_i, y')
	  \,\middle|\,
	  \begin{array}{l}
	     \forall x \in f^{-1}(y)\, \exists x' \in f^{-1}(y') \text{ such that } \\
	     x \xrightarrow{(\sigma_1,t_1)} \cdots \xrightarrow{(\sigma_i,t_i)} \cdots \xrightarrow{(\sigma_n,t_n)} x' \text{ and } j \neq i \implies \sigma_j = \tau
	  \end{array}
	  \right\}\text{.}
	\]
	From the definition of $\beta$ it is immediate to check that for any $\sigma \in \Sigma_\tau$:
	\[
	  (\sigma,f(x')) \in (\beta\circ f)(x) \iff (\sigma,x') \in \Sigma_{!_{\mathbb{T}^\ast}}(\underline{\alpha})(x)
	\]
	and hence that the quotient map $f$ for $R$ extends to a $!_{\mathbb{T}^\ast}$-behavioural morphism for $\underline{\alpha}$.
	
	For the converse implication take $R$ to be a $!_{\mathbb{T}^\ast}$-bisimulation for $\underline{\alpha}$. This means that there is a lax functor $\beta \in [1,\cat{K}]$ whose carrier $Y$ is given by the set of abstract classes of $R$ and such that the quotient map $f\colon X\to Y$ induced by $R$ satisfies:
	\[
	  f \circ \Sigma_{!_{\mathbb{T}^\ast}}(\underline{\alpha}) = \underline{\beta} \circ f
	  \text{.}
	\]
	In particular, for any $\sigma \in \Sigma_\tau$ we have that $(\sigma,f(x')) \in (\underline{\beta} \circ f)(x) \iff (\sigma,x') \in \bigvee_{\vec{t} \in \mathbb{T}^\ast} \underline{\alpha}_{\vec{t}}(x)$. By straightforward definition expansion we conclude that $R$ is a weak time-abstract bisimulation for $\alpha$.
\end{proofatend}

\subsection{Abstraction over branching}
We capture abstraction over branching by replacing the canonical inclusion $\Set \to \kl(\ENA)$ with $\kl(\mathcal{P}) \to \kl(\ENA)$ as $J \colon \cat{J} \to \cat{K}$. Taking the source category of behavioural morphisms to be different from $\Set$ is somewhat akin to the classical approach towards modelling coalgebraic finite trace equivalence in terms of bisimulation for base categories given by the Kleisli categories for a monadic part of the type functor \cite{hasuo07:trace,jacobssilvasokolova2012:cmcs}. Our approach of replacing $\Set \to \kl(\ENA)$ with $\kl(\mathcal{P}) \to \kl(\ENA)$ was successfully used to recover weak trace semantics for non-deterministic automata in terms of weak bisimilarity via saturation in \cite[Example~7.3]{brengos2015:lmcs}. If applied here, we obtain the notions of strong/weak timed/time-abstract language for TTSs as instances of that of $q$-bisimulation for $q$ discussed discussed above.
\begin{proposition}
	\label{thm:nd-ta-language-equivalences}
	For a TTS  $\alpha$, $\underline{\alpha}$ its functor model, and $R$ an equivalence relation on its carrier:
	\begin{enumerate}
	\item 
		$R$ is a timed language equivalence for $\alpha$ iff it is a $id_{\mathbb{T}^\ast}$-language equivalence for $\underline\alpha$;
	\item
		$R$ is a time-abstract language equivalence for $\alpha$ iff it is a $!_{\mathbb{T}}^\ast$-language equivalence for $\underline\alpha$;
	\item 
		$R$ is a weak timed language equivalence for $\alpha$ iff it is a $\varepsilon_{\mathbb{T}}$-language equivalence for $\underline\alpha$;
	\item
		$R$ is a weak time-abstract language equivalence for $\alpha$ iff it is a $!_{\mathbb{T}^\ast}$-language equivalence for $\underline\alpha$;
	\end{enumerate}
\end{proposition}

\begin{proofatend}
Here, we will only prove the last assertion as the remaining ones are proved in an analogous manner (see also the proof of \cref{thm:tts-q-time-abstract-bisimulation,thm:tts-q-weak-timed-bisimulation,thm:tts-q-weak-time-abstract-bisimulation}). Since we will work with lax functors $1\to \cat{K}=\kl(\ENA)$ we will simplify the notation and associate any such assignment $\pi\colon 1\to \cat{K}$ with the endomorphism it is induced by $\pi_0\colon \pi(\ast)\to \ENA(\pi(\ast))$ which satisfies $id\leq \pi_0$ and $\pi_0\circ \pi_0\leq \pi_0$. 

For the functor $\underline{\alpha}\colon \mathbb{T}^\ast\to \cat{K}$ the lax functor 
$\Sigma_!(\underline{\alpha})\colon 1\to \cat{K}=\kl(\ENA)$ is given in terms of a single endomorphism $\Sigma_!(\underline{\alpha})\colon X\to \ENA X$  by:
\begin{align*}
&x\xrightarrow{\sigma}_{\Sigma_!(\underline{\alpha})} y  \text{ whenever } x\xRightarrow{(\sigma,t)}_\alpha y  \text{ and } x\downarrow \text{ whenever } x\xRightarrow{(\tau,t)}_\alpha y \text{ and } y\downarrow \text{ for some }t\in \mathbb{T}.
\end{align*}
A morphism $f\colon X\to \mathcal{P}Y$ is a weak behavioural morphism for $\alpha$ whenever there is a lax functor $\underline{\beta}\colon 1\to \cat{K}$ whose carrier is $Y$ which makes:
$f^\sharp\circ \Sigma_!(\underline{\alpha}) = \underline{\beta}\circ f^\sharp$. This equation turns $f^\sharp$ into a coalgebra homomorphism with the base category $\kl(\mathcal{P})$. By applying the guidelines of \cite[Example 7.3]{brengos2015:lmcs} we get the desired conclusion.
\end{proofatend}

\subsection{Spectrum of equivalences}
By recovering notions of behavioural equivalences as instances of $q$-bisimulation, we can now apply \cref{thm:behavioural-composition-bisimulation,thm:behavioural-inclusion-bisimulation} to organise them by their discriminating power.

\begin{proposition}
  The diagram below describes the spectrum of equivalences for TTSs, ordered from more (bottom) to less (top) discriminating.
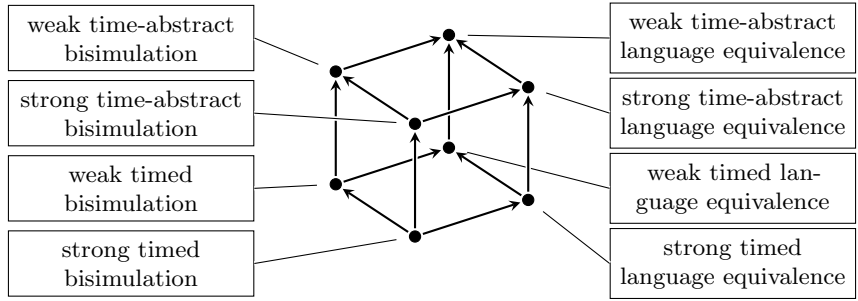
\[
	\centering
	\begin{tikzpicture}[
			auto, font=\footnotesize\upshape,
			scale=1,		
			box/.style={text width=#1, align=center,draw,fill=white},
			dot/.style={circle,draw,fill,inner sep=1.5pt,outer sep=1pt},
			box con/.style={-,thin,shorten >=4pt,crossing=1.7pt},
			hasse con/.style={-stealth,thick,crossing=2pt}
		]

		\node[box={3cm}] (b000) at (-4,-1.5) {strong timed bisimulation};
		\node[box={3cm}] (b100) at (-4,-0.5) {weak timed bisimulation};
		\node[box={3cm}] (b010) at (-4, 0.5) {strong time-abstract bisimulation};
		\node[box={3cm}] (b110) at (-4, 1.5) {weak time-abstract bisimulation};
		\node[box={3cm}] (b001) at ( 4,-1.5) {strong timed language equivalence};
		\node[box={3cm}] (b101) at ( 4,-0.5) {weak timed language equivalence};
		\node[box={3cm}] (b011) at ( 4, 0.5) {strong time-abstract language equivalence};
		\node[box={3cm}] (b111) at ( 4, 1.5) {weak time-abstract language equivalence};

		\begin{scope}[
				yshift=.2cm,
			]
		\begin{scope}[
			z={(0.866cm,0.4cm)}, 
			x={(-0.866cm,0.4cm)}, 
			y={(0cm,1cm)},
			rotate around y=-10,
			scale=.75,
		]
			\foreach \x/\xc in {0/-1,1/1}{
				\foreach \y/\yc in {0/-1,1/1}{
					\foreach \z/\zc in {0/-1,1/1}{
					\node[dot] (p\x\y\z) at (\xc,\yc,\zc) {};
			}}}
		\end{scope}
		\end{scope}

		\foreach \x in {0,1}{
			\foreach \y in {0,1}{
				\foreach \z/\a in {0/east,1/west}{
			\draw[box con] (b\x\y\z.\a) -- (p\x\y\z);
		}}}
				
		\draw[hasse con] (p100) to (p101);
		\draw[hasse con] (p001) to (p101);
		\draw[hasse con] (p101) to (p111);

		\draw[hasse con] (p000) to (p100);
		\draw[hasse con] (p000) to (p010);
		\draw[hasse con] (p000) to (p001);
		\draw[hasse con] (p100) to (p110);
		\draw[hasse con] (p010) to (p110);
		\draw[hasse con] (p010) to (p011);
		\draw[hasse con] (p001) to (p011);
		\draw[hasse con] (p011) to (p111);
		\draw[hasse con] (p110) to (p111);
		
		\draw[box con] (b010.east) -- (p010); %redraw

	\end{tikzpicture}
\]
\end{proposition}

\begin{proof}
The thesis follows from the commuting diagram of monoid morphisms in \cref{fig:morphisms-equivalence-spectrum}, the fact that the canonical inclusion $\Set \to \kl(\ENA)$ factors as $\Set \to \kl(\mathcal{P}) \to \kl(\ENA)$, \cref{thm:tts-q-timed-bisimulation,thm:tts-q-time-abstract-bisimulation,thm:tts-q-weak-timed-bisimulation,thm:tts-q-weak-time-abstract-bisimulation,thm:nd-ta-language-equivalences,thm:behavioural-composition-bisimulation,thm:behavioural-inclusion-bisimulation}. 
\end{proof}

\section{Beyond timed transition systems}
\label{sec:more-applications}

In this section we illustrate the generality of the results presented by listing some representative examples of timed behavioural models fitting our framework besides our running example of timed transition systems. 

\subsection{Quantalic systems}
\label{sec:q-ta}

Quantalic systems are a generalisation of non-deterministic ones where non-deterministic features are extended with some quantitative ones to describe \eg the maximal cost of a computation.

\subsubsection{Model}
\label{sec:q-ta-model}

Let $(\mathcal{Q},\cdot ,1,\leq)$ be a \emph{unital quantale}, \ie a relational structure such that: 
\begin{itemize}
\item $(\mathcal{Q},\cdot,1)$ is a monoid,
\item $(\mathcal{Q},\leq)$ is a complete lattice,
\item arbitrary suprema are preserved by the monoid multiplication.
\end{itemize}
In other words, a unital quantale is a monoid in the category $\Sup$ of join-preserving homomorphisms between complete join semi-lattices.
In the sequel we will often write $\perp_\mathcal{Q}$ or simply $\perp$ for the supremum of the empty set and denote a quantale $(\mathcal{Q},\cdot ,1,\leq)$ by its carrier set $\mathcal{Q}$, provided the associated structure is clear from the context.

\paragraph{Quantale monad.}\looseness=-1
An arbitrary unital quantale $\mathcal{Q}$ gives rise to the monad $\mathcal{Q}^{(-)}$ over $\Set$, called \emph{quantalic monad}. The functor carrying this monad which assigns to any set $X$ the set $\mathcal{Q}^X$ of all functions from $X$ to $\mathcal{Q}$ and to any map $f\colon X\to Y$ the map $\mathcal{Q}^f\colon \mathcal{Q}^X\to \mathcal{Q}^Y$ given by:
\begin{equation*}
	\mathcal{Q}^f(\phi)(y) = \bigvee_{x:f(x)=y} \phi(x)\text{.}
\end{equation*}
Multiplication $\mu$ and unit $\eta$ of $\mathcal{Q}^{(-)}$ are given on each set $X$ by the functions:
\begin{equation*}
	\mu_X(\psi)(x) = \bigvee_{\phi \in \mathcal{Q}^X} \phi(x) \cdot \psi(\phi)
	\qquad\text{and}\qquad
	\eta_X(x)(y) = \begin{dcases}
		1 & \text{if } x = y\\
		\perp & \text{otherwise}
	\end{dcases}
	\text{.}
\end{equation*}
The powerset monad $\mathcal{P}$ is a special case of the above where the chosen quantale is the boolean one.
The quantalic monad can be equipped with left, right double strengths for the structure $(\Set,\times,1)$; these are given on each component as follows:
\begin{gather*}
	\lstr_{X,Y}(x,\phi)(x',y') = \eta_X(x)(x') \cdot \phi(y') \qquad
	\rstr_{X,Y}(\psi,y)(x',y') = \psi(x') \cdot \eta_Y(y)(y') \\
	\dstr_{X,Y}(\psi,\phi)(x,y) = \psi(x) \cdot \phi(y)
	\text{.}
\end{gather*}
The quantalic monad is thus strong and monoidal but not necessary commutative. Composition in $\kl(\mathcal{Q}^{(-)})$ is defined, for any $f$ and $g$ with suitable domain and codomain, as:
\begin{equation*}
	(g \circ f)(x)(z) = \bigvee_{y} f(x)(y) \cdot g(y)(z)\text{.}
\end{equation*}
The tensor of the monoidal structure $(\kl(\mathcal{Q}^{(-)}),\liftKl{\times},1)$ takes every pair of morphisms $f\colon X \to Y$ and $f'\colon X'\to Y'$ to the arrow given by the function:
\begin{equation*}
	(f \mathbin{\liftKl{\times}} f')(x,x')(y,y') = f(x)(y) \cdot f'(x')(y')
	\text{.}
\end{equation*}
The category $\kl(\mathcal{Q}^{(-)})$ is isomorphic to the category $\mathsf{Mat}(\mathcal{Q})$ of $\mathcal{Q}$-valued matrices \cite{rosen:quanta96}. In this sense, $\cat{Rel}$ is a special case of the above where the chosen unital quantale is the boolean one. Another example of quantale is the set of all languages for a given alphabet $\Sigma$ equipped with the pointwise extension of string concatenation as the monoidal structure and set inclusion as the order. More generally, for a monoid $(M,\odot,e)$ equipped with an order $\leq$ that is preserved by $\cdot$, the set $\mathcal{P}_{\downarrow}M$ of downward closed subsets of $M$ carries a unital quantale structure $(\mathcal{P}_{\downarrow}M,\cdot,1,\subseteq)$ where $1$ is the downward cone with cusp $e$:
\[ 
	1 = \{m \in M \mid m \leq e \}
\]
and $\cdot$ is the downward closure of the pointwise extension of $\odot$ to subsets of $M$:
\[
	X \cdot Y = \{m \mid \exists x \in X, \exists y \in Y (m \leq x \odot y)\}
	\text{.}
\]
Positive arithmetic monoids such as natural and real numbers under addition are source of examples of interest for modelling computations with quantitative aspects. 
\begin{example}[Resource use]
Consider $(\mathbb{N},+,0)$ equipped with the natural order. Because the order is total and has $0$ as its least element, $\mathcal{P}_{\downarrow}\mathbb{N}$ is isomorphic to $\mathbb{N} + \{\infty\}$; the resulting quantale structure is that of natural numbers with addition and the natural order extended with infinity.
If we regard numbers attached to computations weighted using this quantale as costs then, the cost of sequential steps is the sum of the cost of each step and the cost of branching points is the maximum of the cost of each branch.
\end{example}

\begin{example}[Least likelihood]
Consider $([0,1],\cdot,1)$ equipped with the inverted natural order on the real interval $[0,1]$. The order is total, has $1$ as its least element, $0$ as its greatest element, and $\mathcal{P}_{\downarrow}[0,1]$ is isomorphic to $[0,1]$.
If we regard numbers attached to computations weighted using this quantale as probability estimates then, the estimate of sequential steps is the product of the probabilities of each step and the cost of branching points is the least (in the natural order) of the estimates for each branch.
\end{example}

\paragraph{Unobservable and accepting moves.}
Consider the monad $\UM{\mathcal{Q}}$ (carried by $\mathcal{Q}^{(\Sigma_\tau \times \Id)}$) obtained by equipping the quantale monad with labelled and unobservable moves as described above. We will refer to this monad as the \emph{quantalic system monad} since coalgebras for its endofunctor $\mathcal{Q}^{(\Sigma_\tau \times \Id)}$ are quantalic labelled transition systems with unobservable moves. These systems are endomorphisms of the Kleisli category of $\UM{\mathcal{Q}}$. In particular, composition in $\kl(\UM{\mathcal{Q}})$ assigns to every $f\colon X \to Y$ and $g\colon Y \to Z$ the composite map:
\begin{equation*}
	(g\circ f)(x)(\sigma,z) = 
		\bigvee\left\{
				f(x)(\sigma_1,y)\cdot g(y)(\sigma_2,z)
				\,\middle|\,
				\{\sigma_1,\sigma_2\} = \{\sigma,\tau\}\text{ and }y \in Y
			\right\} 
	\text{.}
\end{equation*}
The monad $\UM{\mathcal{Q}}$ is equipped with the left, right, and double strengths below:
\begin{gather*}
	\lstr_{X,Y}(x,\phi)(\sigma,x',y') = \eta_X(x)(\tau,x') \cdot \phi(\sigma,y')\\
	\rstr_{X,Y}(\psi,y)(\sigma,x',y') = \psi(\sigma,x')\cdot \eta_Y(y)(\tau,y')\\
	\dstr_{X,Y}(\psi,\phi)(\sigma,x,y) =
		\bigvee\left\{\psi(\sigma_1,x) \cdot \phi(\sigma_2,y) \mid \{\sigma_1,\sigma_2\} = \{\sigma,\tau\}\right\}\text{.}
\end{gather*}
The tensor of the resulting monoidal structure on $\kl(\UM{\mathcal{Q}})$ takes every pair of sets $X$ and $Y$ to their cartesian product $X \times Y$ and every pair of morphisms $f_1\colon X_1 \to Y_1$ and $f_2\colon X_2\to Y_2$ to the arrow given by the function:
\begin{equation*}
	(f_1 \liftKl{\times} f_2)(x_1,x_2)(\sigma,y_1,y_2) = \bigvee\left\{f_1(x_1)(\sigma_1,y_1) \cdot f_2(x_2)(\sigma_2,y_2) \mid \{\sigma_1,\sigma_2\} = \{\sigma,\tau\}\right\}\text{.}
\end{equation*}

Consider the monad $\UAM{\mathcal{Q}}$ (carried by $\mathcal{Q}^{(\Sigma_\tau \times \Id + \{\checkmark\})}$) obtained by equipping the quantale monad with labelled, unobservable moves and accepting as described above. We will refer to this monad as the \emph{quantalic automata monad} since coalgebras for its endofunctor $\mathcal{Q}^{(\Sigma_\tau \times \Id)}$ are quantalic automata with $\varepsilon$-moves (we adopt the same notation used for $\varepsilon$-NA and $\UAM{\mathcal{P}}$). Composition in the Kleisli category for this monad assigns to every $f\colon X \to Y$ and $g\colon Y \to Z$ the composite map:
\begin{equation*}
	(g\circ f)(x)(w)
	= \begin{dcases}
		\bigvee \left\{
				f(x)(\checkmark),
				f(x)(\tau,y)\cdot g(y)(\checkmark) 
				\,\middle|\,
				y \in Y 
			\right\} 
		& \text{if } w = \checkmark\\
		\bigvee\left\{
				f(x)(\sigma_1,y)\cdot g(y)(\sigma_2,z)
				\,\middle|\,
				\{\sigma_1,\sigma_2\} = \{\sigma,\tau\}\text{ and }y \in Y
			\right\} 
		& \text{if } w = (\sigma,z)
	\end{dcases}
\end{equation*}
The monad $\UAM{\mathcal{Q}}$ is equipped the left, right, and double strengths given below:
\begin{align*}
	\lstr_{X,Y}(x,\phi)(w)
	= & \begin{dcases}
			\phi(\checkmark)
		& \text{if } w = \checkmark\\
			\phi(\sigma,y')
		& \text{if } w = (\sigma,x,y')\\
			\perp
		& \text{otherwise}
	\end{dcases}
	\\
	\rstr_{X,Y}(\psi,y)(w)
		= & \begin{dcases}
				\psi(\checkmark)
			& \text{if } w = \checkmark\\
				\psi(\sigma,x')
			& \text{if } w = (\sigma,x',y)\\
				\perp
			& \text{otherwise}
		\end{dcases}
	\\
	\dstr_{X,Y}(\psi,\phi)(w)
		= & \begin{dcases}
				\psi(\checkmark) \vee \phi(\checkmark)
			& \text{if } w = \checkmark\\
				\bigvee\left\{\psi(\sigma_1,x) \cdot \phi(\sigma_2,y) \mid \{\sigma_1,\sigma_2\} = \{\sigma,\tau\}\right\}
			& \text{if } w = (\sigma,x,y)
		\end{dcases}
\end{align*}
The tensor of the resulting monoidal structure on $\kl(\UM{\mathcal{Q}})$ takes every pair of sets $X$ and $Y$ to their cartesian product $X \times Y$ and every pair of morphisms $f_1\colon X_1 \to Y_1$ and $f_2\colon X_2\to Y_2$ to the arrow given by the function:
\begin{equation*}
	(f_1 \liftKl{\times} f_2)(x_1,x_2)(w)
	= \begin{dcases}
			f_1(x_1)(\checkmark) \vee f_2(x_2)(\checkmark)
		& \!\text{if } w\!=\!\checkmark\\
			\bigvee\!\left\{f_1(x_1)(\sigma_1,y_1) f_2(x_2)(\sigma_2,y_2) \mid \{\sigma_1,\sigma_2\}\!=\!\{\sigma,\tau\}\right\}\!
		& \!\text{if } w\!=\!(\sigma,y_1,y_2)
	\end{dcases}
\end{equation*}

\subsubsection{Behavioural equivalences}
\label{sec:q-ta-behavioural-equivalences}

Let $\alpha$ be a timed quantalic system with carrier $X$ and inputs from the alphabet $\Sigma_\tau$. We write $x \xrightarrow{t} \psi$ for a timed step in $\alpha$ \ie:
\begin{align*}
	x \xrightarrow{t} \psi \defiff {} & 
		\psi = \alpha_t(x)\text{.}\\
\intertext{We write $x \xRightarrow{t} \psi$ for a saturated timed step in $\alpha$ \ie:}
	x \xRightarrow{t} \psi \defiff {} & 
		\psi = \lambda(\sigma,x').\bigvee\left\{ \phi(\sigma,x') \,\middle|\, x' \in n < \omega\text{, } X\text{, and } x \xRightarrow{t}_n \phi\right\}\\
\intertext{where the (functional) relation $\xRightarrow{t}_n$ is given by recursion on the number $n$ of underlying steps as follows:}
	x \xRightarrow{t}_{0} \psi \defiff {} & 
		t = 0\text{ and } 
		\psi(\sigma,x') = \begin{cases}
			1 & \text{if } \sigma = \tau \text{ and } x = x'\\
			\perp & \text{otherwise}
		\end{cases}
		\\
	x \xRightarrow{t}_{n+1} \psi \defiff {} &
		\psi = \lambda(\sigma,x').\bigvee\left\{ \phi_1(\sigma_1,x'')\cdot\phi_2(\sigma_2,x') \,\middle|\,
			\begin{array}{l} 
				x'' \in X\text{, } \{\sigma_1,\sigma_2\} = \{\sigma,\tau\} \text{,}\\ 
				x \xrightarrow{t'} \phi_1\text{, and } x'' \xRightarrow{t-t'}_n \phi_2
			\end{array}\right\}\\
\intertext{For an equivalence relation $R$ on the carrier $X$ of $\alpha$, we write $\equiv_R$ for its lifting to $\mathcal{Q}^{\Sigma_\tau \times X}$:}
	\psi \equiv_R \phi \defiff {} &
		 \forall C \in X/R \bigvee_{x \in C} \psi(x) = \bigvee_{x \in C} \phi(x)\text{.}
\end{align*}

\begin{definition}
	\label{def:q-ta-behavioural-equivalences}
	For a timed quantalic system $\alpha$ and an equivalence relation $R$ on its carrier:
	\begin{itemize}
	\item 
		$R$ is a \emph{(strong) timed bisimulation} for $\alpha$ if $x \mathrel{R} x'$ and $x\xrightarrow{t} \psi$ implies that there are $x'$ and $\psi'$ such that $x'\xrightarrow{t} \psi'$ and $\psi \equiv_R \psi'$; 
	\item 
		$R$ is a \emph{(strong) time-abstract bisimulation} for $\alpha$ if $x \mathrel{R} x'$ and $x\xrightarrow{t} \psi$ implies that there are $x'$  and $t'$ such that $x'\xrightarrow{t'} \psi'$ and $\psi \equiv_R \psi'$; 
	\item
		$R$ is a \emph{weak timed bisimulation} for $\alpha$ if $x \mathrel{R} x'$ and $x\xRightarrow{t} \psi$ implies that there are $x'$ and $\psi'$ such that $x'\xRightarrow{t} \psi'$ and $\psi \equiv_R \psi'$; 
	\item
		$R$ is a \emph{weak time-abstract bisimulation} for $\alpha$ if $x \mathrel{R} x'$ and $x\xRightarrow{t} \psi$ implies that there are $x'$  and $t'$ such that $x'\xRightarrow{t'} \psi'$ and $\psi \equiv_R \psi'$;  
	\end{itemize}
\end{definition}

Timed quantalic systems are modelled in the lax functors framework by putting $\cat{K}$ to be $\kl(\UM{\mathcal{Q}})$. Notions of behavioural equivalence for TQSs defined above are all instances of $q$-bisimulation where behavioural morphisms are sourced from $\Set$.

\begin{proposition}
	\label{thm:q-ta-behavioural-equivalences}
	For a timed quantalic system $\alpha$ with functor model $\underline{\alpha}$ and an equivalence relation $R$ on the carrier of $\alpha$:
	\begin{enumerate}
	\item 
		$R$ is a (strong) timed bisimulation for $\alpha$ iff it is a $id_{\mathbb{T}^\ast}$-bisimulation for $\underline\alpha$;
	\item
		$R$ is a (strong) time-abstract bisimulation for $\alpha$ iff it is a $!_{\mathbb{T}}^\ast$-bisimulation for $\underline\alpha$;
	\item 
		$R$ is a weak timed bisimulation for $\alpha$ iff it is a $\varepsilon_{\mathbb{T}}$-bisimulation for $\underline\alpha$;
	\item
		$R$ is a weak time-abstract bisimulation for $\alpha$ iff it is a $!_{\mathbb{T}^\ast}$-bisimulation for $\underline\alpha$;
	\end{enumerate}
\end{proposition}

\subsubsection{Language equivalences}
\label{sec:q-ta-language-equivalences}

Languages accepted by a timed quantalic systems are defined as in the non-deterministic case except that acceptance is no longer a boolean property but a weighted one \ie functions from the set of all timed (resp.~untimed) words to the set carrying the given quantale $\mathcal{Q}$. Let $\alpha$ be a timed quantalic system with accepting moves and let $X$ be its state space. Formally, the timed/untimed languages strongly/weakly accepted by a state $x_0$ of $\alpha$ are the functions:
\begin{align*}
	\mathrm{tl}_{\alpha}(x_0)(t_0\,\sigma_1\,t_1\,\sigma_1\dots t_n) \defeq &
	\bigvee\left\{
		\left(\prod_{i=1}^{n} \psi_i(\sigma_{i},x_{i}) \right) \cdot \psi_{n}(\checkmark)
	\,\middle| %\,
		\begin{array}{l}
		x_1,\dots,x_n \in X\text{ and}\\
		x_{i} \xrightarrow{t_{i}} \psi_{i} \text{ for } i \leq n
		\end{array}
	\right\}
	\\
	\mathrm{utl}_{\alpha}(x_0)(\sigma_1\ldots\sigma_n) \defeq &
	\bigvee\big\{
		\mathrm{tl}_{\alpha}(x_0)(t_0\,\sigma_1\,t_1\dots t_n)
	\,\big|\,
			t_0,\dots,t_n \in \mathbb{T}
	\big\}
	\\
	\mathrm{wtl}_{\alpha}(x_0)(t_0\,\sigma_1\,t_1\,\sigma_1\dots t_n) \defeq &
		\bigvee\left\{
			\left(\prod_{i=1}^{n} \psi_i(\sigma_{i},x_{i}) \right) \cdot \psi_{n}(\checkmark)
		\,\middle| %\,
			\begin{array}{l}
			x_1,\dots,x_n \in X\text{ and}\\
			x_{i} \xRightarrow{t_{i}} \psi_{i} \text{ for } i \leq n
			\end{array}
		\right\}
	\\
	\mathrm{wutl}_{\alpha}(x_0)(\sigma_1\ldots\sigma_n) \defeq &
		\bigvee\big\{
			\mathrm{wtl}_{\alpha}(x_0)(t_0\,\sigma_1\,t_1\dots t_n)
		\,\big|\,
				t_0,\dots,t_n \in \mathbb{T}
		\big\}
	\text{.}
\end{align*}
\begin{definition}
	\label{def:q-ta-language-equivalences}
	For a timed quantalic system $\alpha$ and an equivalence relation $R$ on its carrier:
	\begin{itemize}
	\item 
		$R$ is a \emph{timed language equivalence} for $\alpha$ if $x \mathrel{R} y$ implies that $\mathrm{tl}_{\alpha}(x) = \mathrm{tl}_{\alpha}(y)$; 
	\item 
		$R$ is a \emph{time-abstract language equivalence} for $\alpha$ if $x \mathrel{R} y$ implies that $\mathrm{utl}_{\alpha}(x) = \mathrm{utl}_{\alpha}(y)$; 
	\item
		$R$ is a \emph{weak timed language equivalence} for $\alpha$ if $x \mathrel{R} y$ implies that $\mathrm{wtl}_{\alpha}(x) = \mathrm{wtl}_{\alpha}(y)$; 
	\item
		$R$ is a \emph{weak time-abstract language equivalence} for $\alpha$ if $x \mathrel{R} y$ implies that $\mathrm{wutl}_{\alpha}(x) = \mathrm{wutl}_{\alpha}(y)$. 
	\end{itemize}
\end{definition}

Quantalic systems with accepting moves are modelled in the lax functors framework by putting $\cat{K}$ to be $\kl(\UAM{\mathcal{Q}})$. Notions of language equivalence for TQSs defined above are all instances of $q$-bisimulation where behavioural morphisms are sourced from $\kl(\mathcal{Q}^{(-)})$.
\begin{proposition}
	\label{thm:q-ta-language-equivalences}
	For a timed quantalic system $\alpha$ with functor model $\underline{\alpha}$ and an equivalence relation $R$ on the carrier of $\alpha$:
	\begin{enumerate}
	\item 
		$R$ is a timed language equivalence for $\alpha$ iff it is a $id_{\mathbb{T}^\ast}$-language equivalence for $\underline\alpha$;
	\item
		$R$ is a time-abstract language equivalence for $\alpha$ iff it is a $!_{\mathbb{T}}^\ast$-language equivalence for $\underline\alpha$;
	\item 
		$R$ is a weak timed language equivalence for $\alpha$ iff it is a $\varepsilon_{\mathbb{T}}$-language equivalence for $\underline\alpha$;
	\item
		$R$ is a weak time-abstract language equivalence for $\alpha$ iff it is a $!_{\mathbb{T}^\ast}$-language equivalence for $\underline\alpha$;
	\end{enumerate}
\end{proposition}

It follows from \cref{thm:q-ta-behavioural-equivalences,thm:q-ta-language-equivalences,thm:behavioural-composition-bisimulation,thm:behavioural-inclusion-bisimulation} that notions of bisimulation and language equivalence for TQSs given in \cref{def:q-ta-behavioural-equivalences,def:q-ta-language-equivalences} are organised in the spectrum shown in \cref{fig:timed-equivalences-spectrum}.

\subsection{Segala systems}
\label{sec:s-ta}

Segala systems are systems whose computations express a combination of  non-deterministic and probabilistic aspects. These systems are \citeauthor{segala:phd-thesis} who pioneered their study \cite{segala:phd-thesis}.

\subsubsection{Model}
\label{sec:s-ta-model}

Setting aside for a moment inputs, acceptance, and other aspects typical of automata and focusing on the branching type of Segala systems, the computations expressed by these systems can be ascribed to the functor $\mathcal{P D}$ where $\mathcal{D}$ is the probability distribution monad \cite{sokolova11}. Despite the fact that both functors carry monad structures, their composite does not fail to be a proper monad \cite{jacobs08:cmcs,varaccawinskel2006:mscs}. In order to compose single steps of these machines into sequences and computations, Segala introduced the notion of \emph{combined step} where non-determinism is restricted to those subsets that are closed under convex combinations---hence the term \emph{convex semantics} (of Segala systems). This semantics is captured by the convex combination monad\footnote{
	The convex combinations monad was first introduced in its full generality by \citeauthor{jacobs08:cmcs} in \cite{jacobs08:cmcs} to study trace semantics for combined possibilistic and probabilistic systems. Independently, \citeauthor{brengos2015:lmcs} \cite{brengos2015:lmcs} and \citeauthor{gp:icalp2014} \cite{gp:icalp2014} tweaked Jacobs' construction slightly, so that the resulting monads are more suitable to model  Segala systems and their weak bisimulations. \citeauthor{jacobs08:cmcs}' monad, \citeauthor{brengos2015:lmcs}' monad and \citeauthor{gp:icalp2014}'s monad identify Kleisli categories that are $\Dcpo$-enriched and whose hom-sets admit binary joins. For the purposes of this paper we take the convex combinations monad $\CM\colon \Set\to \Set$ to be that considered in \cite[§8]{brengos2015:lmcs}.
} \cite{jacobs08:cmcs,varaccawinskel2006:mscs,brengos2015:lmcs,gp:icalp2014,ms2017:fi}.

\paragraph{Convex set monad.}
By $[0,\infty)$ we denote the semiring $([0,\infty),+,0,\cdot,1)$ of non-negative real numbers with ordinary addition and multiplication. By a $[0,\infty)$-\emph{semimodule} we mean the commutative monoid with actions $[0,\infty)\times (-)\to (-)$ satisfying axioms listed in \eg \cite{golan:99}. For a set $X$ and define $\mathcal{M} X$ as the set of finitely supported weight functions:
\begin{equation*}
	\mathcal{M} X \defeq \{\phi\colon X\to [0,\infty) \mid \supp(\phi)  \text{ is finite}\}
\end{equation*}
where $\supp(\phi)$ is the \emph{support} of $\phi\colon X \to [0,\infty)$ \ie the set $\{x \mid \phi(x) \neq 0\}$ of elements assigned a non zero value by $\phi$.
We will often denote elements of $\mathcal{M}X$ using the formal sum notation: for $\phi \in \mathcal{M}X$ we write $\sum_x \phi(x) \cdot x$ or, given $\supp(\phi) = \{x_1,\ldots,x_n\}$, simply $\sum_{i=1,\ldots,n} \phi(x_i)\cdot x_i$. Any function $f\colon X\to Y$ induces the action $\mathcal{M} f\colon \mathcal{M} X \to \mathcal{M} Y$ defined as follows:
\begin{equation*}
	 \mathcal{M} (f)(\phi)
	 = \sum \phi(x) \cdot f(x)
	 = \lambda y. \sum_{x\in f^{-1}(y)}\phi(x)
\end{equation*}
It is immediate to check that this data defines the (finite generalised) multiset functor. 
The set $\mathcal{M} X$ carries a monoid structure via pointwise operation of addition, and $[0,\infty)$-action via 
\begin{equation*}
	(a\cdot \phi) (x) \defeq a\cdot \phi(x)\text{,}
\end{equation*}
which turn $\mathcal{M} X$ into a free semimodule over $X$ (see \eg \cite{brengos2015:lmcs, golan:99,jacobs08:cmcs} for details). For a non-empty subset $U\subseteq \mathcal{M} X$ we define its convex closure as the set:
\begin{equation*}
	\overline{U} =
	\left\{ a_1\cdot \phi_1 +\ldots +a_n\cdot \phi_n 
	\;\middle|\; 
	\phi_i\in U,\ a_i\in [0,\infty) \text{ such that } \sum_{i=1}^na_i  =1
	\right\}
	\text{.}
\end{equation*}
Call a subset $U\subseteq \mathcal{M}X$ \emph{convex} provided $\overline{U} = U$ and put
\begin{equation*}
 \CM X = \{U\subseteq \mathcal{M}X \mid U \text{ is convex and non-empty}\}
 \text{.}
\end{equation*}
Addition and $[0,\infty)$-actions are extended to convex sets in a pointwise manner: for $U, V\subseteq \mathcal{M}X$ both convex and $a \in [0,\infty)$ we write $U + V$ for the set $\{\phi + \psi \mid \phi \in U, \psi \in V \}$ and $a \cdot U$ for the set $\{a \cdot \phi \mid \phi \in U \}$.

Every function $f\colon X \to Y$ induces an action $\CM (f)\colon \CM X \to \CM Y$ given on each convex subset $U$ of semimodules on $X$ as follows:
\begin{equation*}
 \CM (f)(U) = \{\mathcal{M}(f)(\phi) \mid \phi \in U\}
 \text{.}
\end{equation*}
The assignment $\CM\colon\Set\to \Set$ is a functor which carries a monadic structure \cite{brengos2015:lmcs} whose multiplication and unit have as component at $X$ the functions:
\begin{equation*}
	\mu_X(U) \defeq \bigcup_{\phi \in U} \sum_{V\in \CM X}\{\phi(V)\cdot \psi \mid \psi \in V\}
	\qquad\text{and}\qquad
	\eta_X(x) \defeq \{1\cdot x\}
	\text{.} 
\end{equation*}
The convex combination monad is equipped with strength, costrength and double strength given on each component as:
\begin{align*}
	\lstr_{X,Y}(x,U) & = \left\{\sum\phi(y)\cdot(x,y)\,\middle|\,\phi \in U\right\}
	\\
	\rstr_{X,Y}(V,y) & = \left\{\sum\psi(x)\cdot(x,y)\,\middle|\,\psi \in V\right\}
	\\
	\dstr_{X,Y}(V,U) & = \left\{\sum(\psi(x)\cdot \phi(y))\cdot(x,y)\,\middle|\,\psi \in V,\phi \in U\right\}	
	\text{.} 
\end{align*}
Composition in $\kl(\CM)$ is given, for each $f$ and $g$ with compatible domains and codomains as:
\begin{equation*}
	(g\circ f)(x) = 
	\bigcup_{\phi \in f(x)}\sum_{y \in \supp(\phi)}
	\left\{ \phi(y)\cdot \psi \mid \psi \in g(y) \right\}
\end{equation*}
where the underlying functions in $\Set$ have type $f\colon X\to \CM(Y)$, $g\colon Y\to \CM(Z)$, and $g\circ f$ $\kl(\CM)$.
The tensor of the symmetric monoidal structure $(\kl(\mathcal{CM}),\liftKl{\times},1)$ takes every pair of morphisms $f\colon X \to Y$ and $f'\colon X'\to Y'$ to the arrow 
\begin{equation*}
	(f \mathbin{\liftKl{\times}} f')(x,x') = \left\{\left\{\sum(\phi(y)\cdot \phi'(y'))\cdot(y,y')\,\middle|\,\phi \in U,\phi' \in U'\right\}\,\middle|\, U \in f(x), U' \in f'(x')\right\}
	\text{.}
\end{equation*}

\paragraph{Unobservable and accepting moves.}
Akin to $\LTS$ and $\ENA$, we extend $\CM$ with unobservable moves and acceptance via the general construction introduced in \cite{brengos2015:lmcs} and introduce the \emph{Segala systems monad} $\ST$ and the \emph{Segala automata monad} $\SA$ over the endofunctors $\CM(\Sigma_\tau \times \Id)$ and $\CM(\Sigma_\tau \times \Id + \{\checkmark\})$, respectively. Since we are interested in the Kleisli categories of these monads, we omit an explicit derivation of their structure and instead we describe directly $\kl(\ST)$ and $\kl(\SA)$. 
The Kleisli category $\kl(\ST)$ has sets as objects and maps $X\to \CM(\Sigma_\tau\times Y)$ as morphisms from $X$ to $Y$. For $f\colon X\to \ST(Y)$ and $g\colon Y\to \ST(Z)$ the composite $g\circ f \in \kl(\ST)$ is the map:
\begin{align*}
	(g\circ f)(x) = {} &
		\bigcup_{\phi \in f(x)}\left(
			\sum_{y \in Y}\{ \phi(\tau,y)\cdot \psi \mid \psi \in g(y)\}
			+{}\right.\\&\left.
			\sum_{\sigma \in \Sigma,\, y \in Y}
			\left\{
				\sum_{i=1}^n \phi(\sigma,y)\cdot r_i \cdot (\sigma,z_i)
			\,\middle|
				\begin{array}{l}
				\sum_{j=1}^{n+m} r_j \cdot (\sigma_j,z_j) 
				\in
				g(y)
				\text{ such that}
				\\
				\sigma_j = \tau \iff j \leq n
				\end{array}
			\right\}\right)
	\\= {} &
		\left\{
			\sum_{\substack{z \in Z, \sigma \in \Sigma_\tau}}
			\left(
			\sum_{y \in Y,\,\{\sigma_1,\sigma_2\} = \{\sigma,\tau\}}\mspace{-12mu}
			\phi(\sigma_1,y)\cdot\psi_y(\sigma_2,z)\right)\cdot(\sigma,z)
		\,\middle|
			\begin{array}{l}
				\phi \in f(x),\\ 
				\{\psi_y\}_{y \in \supp(\phi)},\\
				\psi_y \in g(y)
			\end{array}
		\right\}
	\text{.}
\end{align*}
The above Kleisli category (see \cite{brengos2015:lmcs} for a discussion) lets us consider Segala systems with unobservable moves as endomorphisms and allows their mutual composition. Likewise, the Kleisli category $\kl(\SA)$ has sets as objects and maps $X\to \CM(\Sigma_\tau\times Y + \{\checkmark\})$ as morphisms from $X$ to $Y$. For $f\colon X\to \SA(Y)$ and $g\colon Y\to \SA(Z)$ the composite $g\circ f \in \kl(\SA)$ is the map:
\begin{align*}
	(g\circ f)(x) = {} &
		\bigcup_{\phi \in f(x)}\left(
		\{\phi(\checkmark) \cdot \checkmark\} + 
		\sum_{y \in Y}\{ \phi(\tau,y)\cdot \psi \mid \psi \in g(y)\}
		+{}\right.
		\\ &
		\left.
		\sum_{y \in Y,\, \sigma \in \Sigma}
		\left\{
			\sum_{i=1}^n \phi(\sigma,y)\cdot r_i \cdot (\sigma,z_i)
		\,\middle|
			\begin{array}{l}
			\sum_{j=1}^{n+m} r_j \cdot (\sigma_j,z_j) 
			\in
			g(y)
			\text{ such that}
			\\
			\sigma_i = \tau \iff i \leq n
			\end{array}
		\right\}\right)
	\\ = {}& 
	\left\{
		\sum_{\substack{z \in Z,\ \sigma \in \Sigma_\tau}}
		\left(
			\sum_{y \in Y,\,\{\sigma_1,\sigma_2\} = \{\sigma,\tau\}}\mspace{-12mu}
			\phi(\sigma_1,y)\cdot\psi_y(\sigma_2,z)
		\right) \cdot (\sigma,z)
		+{}\right.\\
		& 
		\left.
		\left(\phi(\checkmark) + 
		\sum_{y \in Y}
		\phi(\tau,y)\cdot\psi_y(\checkmark)\right)\cdot \checkmark
	\,\middle|
		\begin{array}{l}
			\phi \in f(x),\\ 
			\{\psi_y \in g(y)\}_{y \in \supp(\phi)}
		\end{array}
	\right\}
	\text{.}
\end{align*}
Akin to how endomorphisms in $\kl(\ENA)$ model non-deterministic automata, endomorphisms in $\kl(\SA)$ are a model for Segala automata since $\SA$ extends $\ST$ with acceptance mimicking the situation of $\ENA$ and $\LTS$.

\subsubsection{Behavioural equivalences}
\label{sec:s-ta-behavioural-equivalences}

Let $\alpha$ be a timed Segala system with carrier $X$ and inputs from the alphabet $\Sigma_\tau$. We write $x \xrightarrow{t} \psi$ for a timed step in $\alpha$ \ie:
\begin{align*}
	x \xrightarrow{t} \psi \defiff {} & 
		\psi \in \alpha_t(x)\text{.}\\
\intertext{We write $x \xRightarrow{t} \psi$ for a saturated timed step in $\alpha$ \ie:}
	x \xRightarrow{t} \psi \defiff {} & 
		\exists n\, x \xRightarrow{t,n} \psi\\
\intertext{where the relation $\xRightarrow{t,n}$ is given by recursion on the number $n$ of underlying steps as follows:}
	x \xRightarrow{t}_{0} \psi \defiff {} & 
		t = 0 \text{ and } \psi = 1 \cdot (\tau,x)\text{,}\\
	x \xRightarrow{t}_{n+1} \psi \defiff {} &
		\exists\,
		x \xrightarrow{t'} {\textstyle\sum_{i=0}^{m} r_i \cdot (\sigma_i,x'_i)},\ 
		\forall\, i \in \{1,\dots,m\}\,
		\\{} &
		\exists\,
		x'_i \xRightarrow{t - t'}_{n} {\textstyle\sum_{j = 0}^{m_i} r'_{i,j} \cdot (\sigma'_{i,j},x''_{i,j})},\ 
		\forall j \in \{1,\dots,m_i\}\,\exists \sigma''_{i,j} \in \Sigma_\tau
		\\{} &
		\text{ s.t. }
		\{\sigma''_{i,j},\tau\} = \{\sigma_i,\sigma'_{i,j}\}\,\,
		\text{ and }
		\psi =  {\textstyle\sum_{i=0}^{m}\sum_{j = 0}^{m_i} \left(r_i\cdot r'_{i,j}\right) \cdot (\sigma''_{i,j},x''_{i,j})}\text{.}\\
\intertext{For an equivalence relation $R$ on the carrier $X$ of $\alpha$, we write $\equiv_R$ for its $\mathcal{D}(\Sigma_\tau \times \Id)$-extension \ie the equivalence relation on $\mathcal{D}(\Sigma_\tau \times X)$ defined as follows:}
	\psi \equiv_R \phi \defiff {} &
		 \forall C \in X/R,\ \forall \sigma \in \Sigma_\tau\  \sum_{x \in C} \psi(\sigma,x) = \sum_{x \in C} \phi(\sigma,x)\text{.}
\end{align*}
\begin{definition}
	\label{def:s-ta-behavioural-equivalences}
	For a timed Segala system $\alpha$ and an equivalence relation $R$ on its carrier:
	\begin{itemize}
	\item 
		$R$ is a \emph{(strong) timed bisimulation} for $\alpha$ if $x \mathrel{R} x'$ and $x\xrightarrow{t} \psi$ implies that there are $x'$ and $\psi'$ such that $x'\xrightarrow{t} \psi'$ and $\psi \equiv_R \psi'$; 
	\item 
		$R$ is a \emph{(strong) time-abstract bisimulation} for $\alpha$ if $x \mathrel{R} x'$ and $x\xrightarrow{t} \psi$ implies that there are $x'$  and $t'$ such that $x'\xrightarrow{t'} \psi'$ and $\psi \equiv_R \psi'$; 
	\item
		$R$ is a \emph{weak timed bisimulation} for $\alpha$ if $x \mathrel{R} x'$ and $x\xRightarrow{t} \psi$ implies that there are $x'$ and $\psi'$ such that $x'\xRightarrow{t} \psi'$ and $\psi \equiv_R \psi'$; 
	\item
		$R$ is a \emph{weak time-abstract bisimulation} for $\alpha$ if $x \mathrel{R} x'$ and $x\xRightarrow{t} \psi$ implies that there are $x'$  and $t'$ such that $x'\xRightarrow{t'} \psi'$ and $\psi \equiv_R \psi'$;  
	\end{itemize}
\end{definition}

Timed Segala systems are modelled in the lax functors framework by putting $\cat{K}$ to be $\kl(\ST)$. Notions of behavioural equivalence for TSSs defined above are all instances of $q$-bisimulation where behavioural morphisms are sourced from $\Set$.
\begin{proposition}
	\label{thm:s-ta-behavioural-equivalences}
	For a timed Segala system $\alpha$ with functor model $\underline\alpha$ and an equivalence relation $R$ on the carrier of $\alpha$:
	\begin{enumerate}
	\item 
		$R$ is a (strong) timed bisimulation for $\alpha$ iff it is a $id_{\mathbb{T}^\ast}$-bisimulation for $\underline\alpha$;
	\item
		$R$ is a (strong) time-abstract bisimulation for $\alpha$ iff it is a $!_{\mathbb{T}}^\ast$-bisimulation for $\underline\alpha$;
	\item 
		$R$ is a weak timed bisimulation for $\alpha$ iff it is a $\varepsilon_{\mathbb{T}}$-bisimulation for $\underline\alpha$;
	\item
		$R$ is a weak time-abstract bisimulation for $\alpha$ iff it is a $!_{\mathbb{T}^\ast}$-bisimulation for $\underline\alpha$;
	\end{enumerate}
\end{proposition}

\subsubsection{Language equivalences}
\label{sec:s-ta-language-equivalences}

A (timed) language accepted by a (timed) Segala system is a convex set of distributions over (timed) words.
Let $\alpha$ be a timed Segala system with accepting moves and let $X$ be its state space. Formally, the timed/untimed languages strongly/weakly accepted by a state $x$ of $\alpha$ are the convex sets defined below.
\begin{align*}
	\textstyle \sum_{i=1}^{m} r_i\cdot w_i \in \mathrm{tl}_{\alpha}(x) \defiff {}&
		\forall i\ \exists a_0,\dots,a_p,u_0,\dots,u_p \text{ s.t. }  
		\textstyle \sum a_j = 1 \text{, }
		\textstyle \sum a_j u_j = r_i \text{, and}
		\\&
		\text{if } w_i = t \text{ then }
		\forall j \leq p\, \exists \phi \in \ST(X) \text{ s.t. } 
		x \xrightarrow{t} u_j \cdot \checkmark + \phi
		\\&
		\text{if } w_i = t\sigma w' \text{ then } \forall j\leq p\,
		\exists \phi,y_0,\dots,y_q,\psi_0,\dots,\psi_q
		\text{ s.t. }
		\\& \qquad\textstyle
		\forall k\leq q\ \psi_k \in \mathrm{tl}_{\alpha}(y_k)\text{, }
		\forall z \notin \{y_0,\dots,y_q\}\ \phi(\sigma,z) = 0\text{, }
		\\& \qquad\textstyle
		u_j = \sum_{k=0}^{q} \phi(\sigma,y_k)\psi_k(w') \text{, and }
		x \xrightarrow{t} \phi
	\\
	\textstyle \sum_{i=1}^{m} r_i\cdot w_i \in \mathrm{utl}_{\alpha}(x) \defiff {}&
			\forall i\ \exists t_0\sigma_1t_1\dots \sigma_nt_n \in \mathrm{tl}_{\alpha}(x) \text{ s.t. }  
			w_i = \sigma_1\dots\sigma_{n}
	\\
	\textstyle \sum_{i=1}^{m} r_i\cdot w_i \in \mathrm{wtl}_{\alpha}(x) \defiff {}&
		\forall i\ \exists a_0,\dots,a_p,u_0,\dots,u_p \text{ s.t. }  
		\textstyle \sum a_j = 1 \text{, }
		\textstyle \sum a_j u_j = r_i \text{, and}
		\\&
		\text{if } w_i = t \text{ then }
		\forall j \leq p\, \exists \phi \in \ST(X) \text{ s.t. } 
		x \xRightarrow{t} u_j \cdot \checkmark + \phi
		\\&
		\text{if } w_i = t\sigma w' \text{ then } \forall j\leq p\,
		\exists \phi,y_0,\dots,y_q,\psi_0,\dots,\psi_q
		\text{ s.t. }
		\\& \qquad\textstyle
		\forall k\leq q\ \psi_k \in \mathrm{tl}_{\alpha}(y_k)\text{, }
		\forall z \notin \{y_0,\dots,y_q\}\ \phi(\sigma,z) = 0\text{, }
		\\& \qquad\textstyle
		u_j = \sum_{k=0}^{q} \phi(\sigma,y_k)\psi_k(w') \text{, and }
		x \xRightarrow{t} \phi
	\\
	\textstyle \sum_{i=1}^{m} r_i\cdot w_i \in \mathrm{utl}_{\alpha}(x) \defiff {}&
			\forall i\ \exists t_0\sigma_1t_1\dots \sigma_nt_n \in \mathrm{wtl}_{\alpha}(x) \text{ s.t. }  
			w_i = \sigma_1\dots\sigma_{n}
\end{align*}
\begin{definition}
	\label{def:s-ta-language-equivalences}
	For a timed Segala system $\alpha$ and an equivalence relation $R$ on its carrier:
	\begin{itemize}
	\item 
		$R$ is a \emph{timed language equivalence} for $\alpha$ if $x \mathrel{R} y$ implies that $\mathrm{tl}_{\alpha}(x) = \mathrm{tl}_{\alpha}(y)$; 
	\item 
		$R$ is a \emph{time-abstract language equivalence} for $\alpha$ if $x \mathrel{R} y$ implies that $\mathrm{utl}_{\alpha}(x) = \mathrm{utl}_{\alpha}(y)$; 
	\item
		$R$ is a \emph{weak timed language equivalence} for $\alpha$ if $x \mathrel{R} y$ implies that $\mathrm{wtl}_{\alpha}(x) = \mathrm{wtl}_{\alpha}(y)$; 
	\item
		$R$ is a \emph{weak time-abstract language equivalence} for $\alpha$ if $x \mathrel{R} y$ implies that $\mathrm{wutl}_{\alpha}(x) = \mathrm{wutl}_{\alpha}(y)$. 
	\end{itemize}
\end{definition}

Timed Segala systems with accepting moves are modelled in the lax functors framework by putting $\cat{K}$ to be $\kl(\SA)$. Notions of language equivalence for TSSs defined above are all instances of $q$-bisimulation where behavioural morphisms are sourced from $\kl(\CM)$.
\begin{proposition}
	\label{thm:s-ta-language-equivalences}
	For a timed Segala system $\alpha$ with functor model $\underline\alpha$ and an equivalence relation $R$ on the carrier of $\alpha$:
	\begin{enumerate}
	\item 
		$R$ is a timed language equivalence for $\alpha$ iff it is a $id_{\mathbb{T}^\ast}$-language equivalence for $\underline\alpha$;
	\item
		$R$ is a time-abstract language equivalence for $\alpha$ iff it is a $!_{\mathbb{T}}^\ast$-language equivalence for $\underline\alpha$;
	\item 
		$R$ is a weak timed language equivalence for $\alpha$ iff it is a $\varepsilon_{\mathbb{T}}$-language equivalence for $\underline\alpha$;
	\item
		$R$ is a weak time-abstract language equivalence for $\alpha$ iff it is a $!_{\mathbb{T}^\ast}$-language equivalence for $\underline\alpha$;
	\end{enumerate}
\end{proposition}

It follows from \cref{thm:s-ta-behavioural-equivalences,thm:s-ta-language-equivalences,thm:behavioural-composition-bisimulation,thm:behavioural-inclusion-bisimulation} that notions of bisimulation and language equivalence for TSSs given in \cref{def:s-ta-behavioural-equivalences,def:s-ta-language-equivalences} are organised in the spectrum shown in \cref{fig:timed-equivalences-spectrum}.

\subsection{Fully probabilistic processes and beyond}\label{sec:fully-probabilistic-processes-and-beyond}
The aim of this subsection is to show how probabilistic systems (or more generally, weighted systems \cite{brengos2015:jlamp}) fit into the framework presented in this paper. Weighted systems, unlike other systems presented in this section tend not to satisfy one of the main assumptions of Theorem \ref{thm:general-saturation}, namely, left distributivity. There are two (essentially equivalent) workarounds to this problem known in the literature. One is to consider the continuous continuation monad from \cite{gp:icalp2014} which is a natural extension of the weighted systems type monad and whose Kleisli category satisfies left distributivity (called \emph{algebraicity} in \cite{gp:icalp2014}). The idea behind the second solution, which is given in \cite{brengos2015:jlamp}, is similar and formulated on the categorical level: given a suitably order enriched category $\cat{K}$ which does not satisfy left distributivity we can embed it into $\widehat{\cat{K}}$, a category with left distributivity (which admits saturation). Although, the latter proposal seems to be more general than the former, the first solution preserves the coalgebraic nature of systems, \ie systems of a type $X\to TX$ are extended to systems of a type $X\to T'X$, where $T'$ is the new, richer monad. The second embedding, if applied to $\cat{K}=\kl(T)$ does not necessarily yield a category $\widehat{\cat{K}}$ which is a Kleisli category for some monad $T'$. However, as shown in \cite{brengos2015:jlamp}, if we adopt the general embedding $\cat{K}\hookrightarrow \widehat{\cat{K}}$ approach for saturation admittance, weak behavioural equivalence for systems in $\cat{K}$ (as defined in \loccit) remains the same regardless of the choice of the embedding. In other words, it is enough to embed $\cat{K}$ into \emph{a} left distributive category (with saturation) and not into any particular choice of the left distributive category\footnote{This also means that if $\cat{K}$ is left distributive and satisfies all other required properties then we can embed it into itself. Hence, using the guidelines of \cite{brengos2015:jlamp}, it can be easily shown that the setting we present below in \cref{subsection:embedding_presheafs} generalizes the one from \cref{sec:general-saturation-and-equivalences}. }. See \cite{brengos2015:jlamp} for a discussion.

We will recall some of the ingredients of the construction of a left distributive category $\widehat{\cat{K}}$ from \cite{brengos2015:jlamp} and elaborate more on how can this be used in order to define $q$-saturations and $q$-behavioural morphisms. We start off with a general construction and then, in the next paragraph, we instantiate it on probabilistic processes.
\subsubsection{The embedding}\label{subsection:embedding_presheafs}
The original construction of $\widehat{\cat{K}}$ was presented in the context of $\Cpoj$-enriched categories. By $\Cpoj$ we denote the category of posets which admit binary joins and countable joins $\bigvee_{n<\omega}x_n$ of $\omega$-chains $x_1\leq x_2\leq \ldots$ with morphisms preserving joins of $\omega$-chains. Hence, a category is $\Cpoj$-enriched if each hom-set is a poset which admits binary joins and joins of $\omega$-chains, and the composition preserves the latter.  Let $\cat{K}$ be a small\footnote{Smallness of $\cat{K}$ is required to guarantee $\widehat{\cat{K}}$ is locally small \cite{brengos2015:jlamp}. However, our aim is to apply this approach for $\cat{K}=\mathcal{K}l(T)$ where $T$ is a monad for weighted systems. Unfortunately, it is not a \emph{small} category. The simplest solution in this case is to take $\cat{K}$ as a suitable full subcategory of
$\mathcal{K}l(T)$ that meets our requirements. For example, if we are
interested in $T$-coalgebras whose base category is $\cat{Set}$ and
which have a carrier of cardinality below $\kappa$ then we can take 
\cat{K} to have exactly one set of cardinality $\lambda$ for any $\lambda<\kappa$. 
In particular, if $\kappa = \omega$ then \cat{K} is the dual category to the
Lawvere theory for $T$ (see \eg \cite{hyland:power:2007}). See \cite{brengos2015:jlamp} for a thorough discussion.
} $\Cpoj$-enriched category. Following \cite{brengos2015:jlamp} we define $\widehat{\cat{K}}$ to be the opposite of the category $[\cat{K}, \Cpoj]^{\Cpoj}$  of lax functors and oplax natural transformations between them. Taking any object $X$ to its representable functor $\widehat{X}\defeq \cat{K}(X,-)$ and any morphism $f\colon X\to Y$ to $\widehat{f}\defeq \cat{K}(f,-)$ defines the embedding $\widehat{(-)}\colon \cat{K}\to \widehat{\cat{K}}$ into an $\Cpoj$-enriched category which is left distributive \cite{brengos2015:jlamp}. If, in the above, we replace $\Cpoj$ with $\Dcpoj$, \ie the category of complete orders which admit binary joins with directed suprema preserved by the morphisms, then we also get a true statement (\emph{conf. loc. cit.}). In this case, as noted in the footnote under Theorem \ref{thm:general-saturation}, the category $\widehat{\cat{K}}$ satisfies the assumptions of this theorem. 

Here, we need to point out an important property of the embedding $\widehat{(-)}\colon \cat{K}\to\widehat{\cat{K}}$ above. Namely, for fixed objects $X,Y\in \cat{K}$ the hom-poset restriction $\widehat{(-)}\colon \cat{K}(X,Y)\to \widehat{\cat{K}}(\widehat{X},\widehat{Y})$ of $\widehat{(-)}\colon \cat{K}\to \widehat{\cat{K}}$ admits a left order adjoint $\Theta\colon \widehat{\cat{K}}(\widehat{X},\widehat{Y})\to \cat{K}({X},{Y})$ which is given by \cite{brengos2015:jlamp}
\[\Theta(\phi:\widehat{X}\to \widehat{Y}\in \widehat{\cat{K}})=\Theta(\phi:\cat{K}(Y,-)\to \cat{K}(X,-))\defeq \phi_Y(id_Y)\text{.}\]
  
We are now ready to define $q$-bisimulations for $\pi\in [\cat{M},\cat{K}]^J$.  
\subsubsection{$q$-saturations and $q$-bisimulations}
Assume $\cat{K}$ is $\Dcpoj$-enriched and $\cat{J}$ is a subcategory of $\cat{K}$. Take $q\colon M\to N$ an onto monoid homomorphism, and consider $\widehat{J}\colon \cat{J}\to \cat{K}\to \widehat{\cat{K}}$. By Theorem \ref{thm:general-saturation} and the above remarks, $\widehat{\cat{K}}$ admits $q$-saturation \wrt $\widehat{J}$:
	\begin{equation*}
		\begin{tikzpicture}[diagram]
			\node (n0) {$[\cat{M},\widehat{\cat{K}}]^{\widehat{J}}$};
			\node[right=6ex of n0] (n1) {$[\cat{N},\widehat{\cat{K}}]^{\widehat{J}}$};
			\draw[->,bend left] (n0.north east) to node {$\Sigma_q^{\widehat{J}}$} (n1.north west);
			\draw[->,bend left] (n1.south west) to node{$[q,\widehat{\cat{K}}]^{\widehat{J}}$} (n0.south east);
			\node[rotate=-90] at ($(n0.east)!.5!(n1.west)$) {$\dashv$};
		\end{tikzpicture}
	\end{equation*} 
  Assume $\pi \in [\cat{M},\cat{K}]^J$. A \emph{$q$-behavioural morphism} on $\pi$ is any morphism $f$ in $\cat{J}$ with domain $\Sigma_q^{\widehat{J}} (\widehat{\pi})(\ast)$ such that there is $\pi'_n \in [\cat{N},\cat{K}]^J$ making
\begin{equation}
\Theta\left (\widehat{J(f)}\circ \Sigma_q^{\widehat{J}}(\widehat{\pi})_n\right)  = \Theta(\widehat{\pi'}_n\circ \widehat{J(f)})=\pi_n'\circ J(f).\label{equation:strong_equ_morphism}
\end{equation}
A \emph{$q$-bisimulation} on $\pi$ is a kernel pair of a $q$-behavioural morphism whose domain is $\pi$. 

This generalization goes along the lines of the generalization of weak behavioural equivalence via saturation proposed by us in \cite{brengos2015:jlamp}. The remaining part of this subsection will focus on instantiating it on Markov chains. 

\subsubsection{Example: Markov chains}
At first, we present basic categorical ingredients needed to model Markov chains in the lax functorial setting. 
Consider the $\Set$ endofunctor $\mathcal F_{[0,\infty]}$ given on any set $X$ and on any function $f \colon X \to Y$ as
follows:
\[
	\mathcal F_{[0,\infty]}X \defeq
	\{\phi \colon X \to {[0,\infty]} \mid |\{x \mid \phi(x) \neq 0\}| \leq \omega \}
	\qquad
	\mathcal F_{[0,\infty]}f(\phi)(y) \defeq \sum_{x \in f^{-1}(y)\}} \phi(x) \text{.}
\]
This functor extends to a monad whose multiplication
$\mu$ and unit $\eta$ are given on their components by:
\[
	\mu_X(\phi)(x) \defeq \sum_{\psi \in \mathcal F_{[0,\infty]} X} \phi(\psi) \cdot \psi(x)	
	\qquad
	\eta_X(x)(x') \defeq
	\begin{cases}
		1 & \text{ if $x = x'$} \\
		0 & \text{ otherwise.}
	\end{cases}
\]
For any $X,Y \in \Set$ and $f,g\in \kl(\mathcal{F}_{[0,\infty]})(X,Y)$ define:
\[
	f \leq g \defiff  f(x)(y) \leq g(x)(y) \text{ for any } x\in X,\ y \in Y.
\]
As pointed out in \cite{brengos2015:jlamp} the category $\kl(\mathcal{F}_{[0,\infty]})$ is $\Cpoj$-enriched. It is \emph{not} left distributive but it is $\Set$-right distributive \wrt countable  suprema \cite{brengos2015:jlamp}. The category $\widehat{\cat{K}}$ we embed (a full subcategory of) $\kl(\mathcal{F}_{[0,\infty]})$  into by following the guidelines of \cite{brengos2015:jlamp} admits countable suprema, is left distributive and $\cat{Set}$-right distributive \wrt them. However, it does not admit \emph{arbitrary} non-empty suprema. Yet, if we carefully analyse the proof of Theorem \ref{thm:general-saturation} then it turns out that the above assumptions suffice for $!_{\mathbb{N}}\colon \mathbb{N}\to 1$-saturation $\Sigma_{!_\mathbb{N}}: [\mathbb{N},\widehat{\cat{K}}]^{\Set}\to [1,\widehat{\cat{K}}]^{\Set}$ to exist\footnote{Indeed, in this case only countable suprema are considered in the definition of $\Sigma_{!_\mathbb{N}}$.}.

Consider any ordinary functor $\pi\colon \mathbb{N}\to \kl(\mathcal{F}_{[0,\infty]})$ in $[\mathbb{N},\kl(\mathcal{F}_{[0,\infty]})]^{\Set}$ and a $\Set$-map $f$ whose domain is $\pi(\ast)$.  We have the following statement (see also \cite{brengos2015:corr}).
\begin{lemma}
\begin{equation}
\Theta(\widehat{f^\sharp}\circ \Sigma_{!_{\mathbb{N}}}(\widehat{\pi})) = \mu x. (f^\sharp \vee x\circ \pi_1),\label{equation:simpler}
\end{equation} 
with the least fixpoint calculated in $\kl(\mathcal{F}_{[0,\infty]})$ and $(-)^\sharp:\cat{Set}\to \kl(\mathcal{F}_{[0,\infty]})$ the inclusion functor. 
\end{lemma}
\begin{proof}
Put  $\Pi = \bigvee_{n<\omega} \widehat{\pi}_n$. By Theorem \ref{thm:general-saturation} we have:
\begin{align*}
&(\widehat{f^\sharp}\circ \Sigma_{!_{\mathbb{N}}}(\widehat{\pi}))= \Theta(\widehat{f^\sharp}\circ \bigvee_l (\Pi)^l) = \Theta(\bigvee_l \widehat{f^\sharp}\circ  (\Pi)^l) \stackrel{\dagger}{=}\bigvee_l \Theta(\widehat{f^\sharp}\circ  (\Pi)^l) \stackrel{\diamond}{=}\\ &\bigvee_l G^l(f^\sharp) \stackrel{\Box}{=} \mu x.(f^\sharp\vee x\circ \pi_1).
\end{align*}
The equation $(\dagger)$ follows from the fact that $\Theta$ preserves arbitrary suprema (as it is a left adjoint). If we put  $G(x) = f^\sharp\vee \bigvee_{n} x\circ \pi_n$ then $(\diamond)$ holds, which follows by induction. Indeed, for $l=0$ we vacuously have: $G^0(f^\sharp) = \Theta(\widehat{f^\sharp}\circ (\Pi)^0)$. Assume the identity $G^l(f^\sharp) = \Theta(\widehat{f^\sharp}\circ (\Pi)^l)$ holds for $l$ and consider:
\begin{align*}
&\Theta(\widehat{f^\sharp}\circ (\Pi)^{l+1})=\Theta(\widehat{f^\sharp}\circ (\Pi)^{l}\circ \Pi ) = \Theta(\widehat{f^\sharp}\circ (\Pi)^{l}\circ \bigvee_{n<\omega} \widehat{\pi_n} ) = \bigvee_{n<\omega} \Theta(\widehat{f^\sharp}\circ (\Pi)^{l}\circ  \widehat{\pi_n} )= \\
&\bigvee_{n<\omega} \Theta(\widehat{f^\sharp}\circ (\Pi)^{l}) \circ {\pi_n}=\bigvee_n G^l(f^\sharp)\circ \pi_n = f^\sharp\vee \bigvee_n G^l(f^\sharp)\circ \pi_n = G^{l+1}(f^\sharp). 
\end{align*}
Finally, in order to see $(\Box)$ holds consider $H(x)= f^\sharp\vee x\circ \pi_1$.  Obviously $H(x)\leq G(x)$. However, since $\pi$ is an ordinary functor, we also have: 
\[
G(x)=f^\sharp\vee \bigvee_n x\circ (\pi_1)^n\text{.}
\]
Hence, if we take $G_m(x) =f^\sharp\vee \bigvee_{n=1}^m x\circ (\pi_1)^n$ then by induction $G_m^n(f^\sharp) \leq H^{m\cdot n}(f^\sharp)$.  Therefore,
\[\bigvee_{n,m} G^n_m(f^\sharp)\leq \bigvee_n H^n(f^\sharp)\leq \bigvee_n G^n(f^\sharp)\text{.}\]
Since $\bigvee_n H^n(f^\sharp) = \mu x. (f^\sharp\vee x\circ \pi_1)$, as our category is $\Cpoj$-enriched, this proves the assertion.
\end{proof}
\noindent Hence, the identity (\ref{equation:strong_equ_morphism}) becomes:
$
 \mu x. (f^\sharp\vee x\circ \pi_1) = \pi'_1\circ f^\sharp, 
$
where the least fix point and the composition are calculated in $\kl(\mathcal{F}_{[0,\infty]})$.

We are ready to elaborate more on a Markov chain example.  We will now recall some notions from Markov chain theory. The reader is referred to \eg \cite{books/daglib/0095301} for basic definitions and properties.

Let $(X_n)_{n\leq \omega}$ be an Markov chain (or MC in short). We call the chain $(X_n)$ \emph{homogeneous} whenever 
$\mathbb{P}(X_n = j \mid X_m=i) = \mathbb{P}(X_{n-m}=j \mid X_0=i)$.
Any homogeneous MC $(X_n)$ on a finite state space $S$ gives rise to its \emph{transition matrix family}, \ie a family  $\{P(n):S^2\to [0,1]\}_{n\geq 0}$ whose $ij$-th entry $p_{ij}(n)=P(n)(i,j)$ describes the conditional \emph{transition probabilities}:
\[
p_{ij}(n) = \mathbb{P}(X_n = j \mid X_0=i)\text{.}
\] 
The transition matrix family satisfies $P(0)=I$ and $P(m+n)=P(m)\cdot P(n)$, where $I$ is the identity matrix and $\cdot$ is the matrix multiplication. Hence, $P(n) = P(1)^n$. The family $\{P(n)=P(1)^n\}_{n < \omega}$ yields an assignment $\pi:\mathbb{N}\to \kl(\mathcal{F}_{[0,\infty]})$ given for any $n\in \mathbb{N}$ by:
\[
\pi(\ast) = S, \quad \pi_n\colon S\to \mathcal{F}_{[0,\infty]} S; \pi_n(i)(j)= p_{ij}(n)
\text{.}\]
The assignment $\pi = (\pi_n)$ is an ordinary functor $\mathbb{N}\to \kl(\mathcal{F}_{[0,\infty]})$ and, hence, is a member of $[\mathbb{N},\kl(\mathcal{F}_{[0,\infty]})]^\Set$ and will be referred to as the \emph{transition functor} of the  chain $(X_n)$. 
  
Consider an equivalence relation $R$ on the state space $S$. For an abstract class $C$ of $R$ let us denote:
\[
p_{i,C}^n=\mathbb{P}(X_m\in C \text{ for some } m\geq n\mid X_0=i) \text{ and } p_{i,C} = p_{i,C}^0.
\]
\begin{lemma}\label{lemma:homogen_chain}
The family $\{p_{i,C}\}_{i\in S}$ satisfies:
\begin{align}
p_{i,C}= \left \{\begin{array}{cc} 1 & \text{ if } i\in C,\\ \sup_{n\geq 0} \sum_{j\in S} p_{j,C}\cdot p_{i,j}(n) & \text{ otherwise.}  \end{array}\right.
\end{align}
\end{lemma}
\begin{proof}
It is clear that if $i\in C$ then $p_{i,C}=1$. For $i\notin C$ we have:
\begin{align*}
&p_{i,C} = \sup_{n\geq 0} p_{i,C}^n = \sup_{n\geq 0} \sum_{j\in S}\mathbb{P}(X_m\in C, m\geq n\mid X_n=j,X_0=i) \cdot \mathbb{P}(X_n=j\mid X_0=i) \stackrel{\dagger}{=}\\
&\sup_{n\geq 0} \sum_{j\in S}\mathbb{P}(X_m\in C, m\geq n\mid X_n=j) \cdot \mathbb{P}(X_n=j\mid X_0=i) \stackrel{\dagger\dagger}{=}\\
&\sup_{n\geq 0} \sum_{j\in S}\mathbb{P}(X_m\in C, m\geq 0\mid X_0=j) \cdot \mathbb{P}(X_n=j\mid X_0=i) =\sup_{n\geq 0} \sum_{j\in S}p_{j,C}\cdot p_{i,j}(n).
\end{align*}
The identity $(\dagger)$ follows by $(X_n)$ being markovian and $(\dagger\dagger)$ by homogeneity of the given process.
\end{proof}

\begin{theorem}
 The relation $R\subseteq S\times S$ is a $!\colon \mathbb{N}\to 1$-bisimluation on the transition functor $\pi$  of a homogeneous MC $(X_n)_{t\geq 0}$ provided that for any $(i,j)\in R$ and any  abstract class $C$ of $R$ we have:
\begin{equation}
\mathbb{P}(X_n\in C \text{ for some } n\geq 0\mid X_0=i) = \mathbb{P}(X_n\in C \text{ for some } n\geq 0\mid X_0=j).\label{id:equal_probab}
\end{equation}
\end{theorem}

\begin{proof}
Let $f\colon S\to S_{/R}; i\mapsto [i]_{/R}$. By \cref{lemma:homogen_chain} and (\ref{equation:simpler}) we have: $\mu x.(f^\sharp \vee  x\circ \pi_1)(i)(C) = p_{i,C}$. Satisfaction of the identity (\ref{id:equal_probab}) is equivalent to existence of a map $\beta\colon S_{/R}\to \mathcal{F}_{[0,\infty]}S_{/R}$ which makes $\mu x.(f^\sharp \vee  x\circ \pi_1) = \beta \circ f^\sharp$ hold. This proves the assertion.
\end{proof}

\section{Conclusion}
\label{sec:conclusion}

\looseness=-1
In this paper we introduced a general definition of behavioural models with explicit time flow and their behavioural theory. The framework is based on (lax) functors over order-enriched categories, typically Kleisli categories. This approach allows us to encode in the index category how computations are observed (\eg if time durations are associated to single steps or entire computations) while abstracting from other computational aspects which are modelled in the base category. 
A key advantage of this separation is that we can fit into our setting many models of interest using Kleisli categories for standard monads like: the powerset monad, the quantalic monad, the convex set monad, the generalised multiset monad, and many more (see \cite{brengos2015:jlamp,brengos2015:lmcs} for more examples of compatible monads).

Although the categories induced by these monads do not necessarily satisfy our main assumptions, namely left distributivity, our framework is still applicable: we have shown that one only needs an embedding into category a category that satisfies our main assumptions.
As an example we applied this approach to weighted transition systems thus covering fully probabilistic systems and discrete Markov chains, among others.
This technique builds on our previous work on saturation and weak bisimulation: in \cite{brengos2015:jlamp} we provided a general method for constructing such embeddings and identified the class embeddings compatible with saturation.
Notably, these results apply also to categories different from $\Set$; in \loccit we have also considered presheaves, compact Hausdorff spaces and measurable spaces.

Our results are built on lax functors and the theory of saturation which we introduced and developed in a series of works \cite{brengos2015:corr,brengos2015:jlamp,brengos2015:lmcs} that provided a general coalgebraic characterisation of weak behavioural equivalences covering many types of systems of interest.
In this paper (and its conference version, \cite{brengos2016:concur}) we extended the theory of saturation and saturation-based behavioural equivalence developing the notion of \emph{general saturation} for lax functors on a monoid category and \emph{$q$-behavioural equivalence}.

Our framework provides a rich behavioural theory that encompasses wide range of behavioural equivalences found in the literature: we have shown that timed bisimulation, timed language equivalence, as well as their weak and time-abstract counterparts, are all instances of $q$-behavioural equivalence.
Moreover, we proved that all these notions of equivalences are naturally organised by their discriminating power to form a spectrum (\cref{fig:timed-equivalences-spectrum}) and that this result does not depend on the type of the systems under scrutiny. 

% Related work -----------------------------------------------------------------
This is not the first work to consider timed behavioural models and their behavioural theory from a categorical perspective. 
In \cite{gribovskaya:tamc2010}, \citeauthor{gribovskaya:tamc2010} present a categorical view of timed weak bisimulation. Their approach is based on open-maps bisimulation and is limited to non-deterministic timed transition systems.
In \cite{kick:phdthesis}, \citeauthor{kick:phdthesis} presents a coalgebraic framework for behavioural models that combine timed transitions with discrete ones: time-dependent computations are modelled by a suitable comonad over \Set and then combined with other behavioural aspects by means of comonad products.
The approach is inherently limited to systems that can be modelled in \Set, that distinguish between timed and discrete actions (thus excluding timed automata and timed CSP), and to strong timed bisimulation.
The last observation is the main distinguishing point between \loccit and this work: because of technical difficulties associated with comonad products, this approach appears less flexible then ours when behavioural equivalences beside strong timed bisimulation are considered. 

% Future work ------------------------------------------------------------------
The categorical characterization of timed behavioural models paves the way for further interesting lines of research. 
One line is to extend our framework to support other index categories besides monoids. This would allow us to study more structured state spaces and computations like \eg those found in alternating games.
Another line is to extend the framework with results from the rich theory of coalgebras such as $\omega$-behaviours \cite{brengos2018:concur}, minimization \cite{bonchi2014:tcl-mini}, determinisation \cite{bonchi2013:lmcs-det}, and up-to techniques \cite{bonchi2014:lics-upto}.

%-------------------------------------------------------------------------------
%	BACK MATTER
%-------------------------------------------------------------------------------
\bibliography{biblio}

%-------------------------------------------------------------------------------
%	APPENDIX
%-------------------------------------------------------------------------------

\appendix

\omittedproofs[\section{Omitted proofs}\label{proof:proofs}]

\end{document}